\documentclass[a4paper,UKenglish,cleveref, autoref, thm-restate]{LIPICS/lipics-v2021}
\pdfoutput=1 %uncomment to ensure pdflatex processing (mandatatory e.g. to submit to arXiv)
\hideLIPIcs  %uncomment to remove references to LIPIcs series (logo, DOI, ...), e.g. when preparing a pre-final version to be uploaded to arXiv or another public repository

\usepackage{mathtools}
\usepackage{booktabs}
\usepackage{bbm} % for \mathbbm{1} (indicator function)
\usepackage[linesnumbered,ruled,vlined]{algorithm2e}
\SetKwInOut{Input}{Input}
\SetKwInOut{Output}{Output}

\title{Constant-Factor Approximations for Doubly Constrained Fair \texorpdfstring{\(k\)}{k}-Center, \texorpdfstring{\(k\)}{k}-Median and \texorpdfstring{\(k\)}{k}-Means} % working title :)
\titlerunning{Constant-Factor Approximations for Doubly Constrained Fair \texorpdfstring{\(k\)}{k}-Clustering} %TODO optional, please use if title is longer than one line

\author{Nicole Funk}{Department of Informatics, University of Cologne, Germany}{funk@cs.uni-koeln.de}{}{Funded by the Deutsche Forschungsgemeinschaft (DFG, German Research Foundation) – Project Number 559931366}

\author{Annika Hennes}{Heinrich Heine University Düsseldorf, Faculty of Mathematics and Natural Sciences, Germany}{annika.hennes@hhu.de}{}{Funded by the Deutsche Forschungsgemeinschaft (DFG, German Research Foundation) - Project 456558332}

\author{Johanna Hillebrand}{Heinrich Heine University Düsseldorf, Faculty of Mathematics and Natural Sciences, Germany}{johanna.hillebrand@hhu.de}{}{Funded by the Deutsche Forschungsgemeinschaft (DFG, German Research Foundation) - Project 459420781}

\author{Sarah Sturm}{Department of Informatics, University of Bonn, Germany}{ssturm@uni-bonn.de}{}{Funded by the Deutsche Forschungsgemeinschaft (DFG, German Research Foundation) – Project Number
459420781}

\authorrunning{N. Funk, A. Hennes, J. Hillebrand and S. Sturm} %TODO mandatory. First: Use abbreviated first/middle names. Second (only in severe cases): Use first author plus 'et al.'

\Copyright{Nicole Funk, Annika Hennes, Johanna Hillebrand and Sarah Sturm} %TODO mandatory, please use full first names. LIPIcs license is "CC-BY";  http://creativecommons.org/licenses/by/3.0/

 %TODO mandatory; please add comma-separated list of keywords

\supplement{An implementation of the doubly constrained fair k-center algorithm}%optional, e.g. related research data, source code, ... hosted on a repository like zenodo, figshare, GitHub, ...
\supplementdetails{Software}{https://github.com/FairClusteringResearch/doubly_fair_clustering} %linktext, cite, and subcategory are optional

\acknowledgements{We thank Melanie Schmidt for her valuable feedback and advice, Anna Arutyunova for proofreading the manuscript, and Lukas Drexler and Elmar Langetepe for productive initial discussions. We are also grateful to the three anonymous reviewers whose constructive comments helped improve this paper.}%optional

\nolinenumbers %uncomment to disable line numbering

\EventEditors{Pierre Fraigniaud}
\EventNoEds{1}
\EventLongTitle{20th Scandinavian Symposium on Algorithm Theory (SWAT 2026)}
\EventShortTitle{SWAT 2026}
\EventAcronym{SWAT}
\EventYear{2026}
\EventDate{June 17--19, 2026}
\EventLocation{Copenhagen, Denmark}
\EventLogo{}
\SeriesVolume{370}
\ArticleNo{7}
\usepackage{tikz}
\usetikzlibrary{positioning}
\usetikzlibrary{backgrounds,fit,shapes.misc,shapes.geometric,hobby}
\usetikzlibrary{arrows.meta, calc}
\usepackage{tkz-euclide}

\tikzset{
	every picture/.style={line width=1pt},
	point/.style args={#1}{circle, draw=#1, fill=#1, minimum size=4pt, inner sep=0pt},
	center/.style args={#1}{diamond, draw=#1, fill=#1, line width= .5, minimum size=6.5pt, inner sep=0pt},
	outlier/.style args={#1}{circle, draw=#1, minimum size=4pt, inner sep=0pt},
    assign/.style={->,shorten >=.5pt, shorten <= .5pt},
    reassign/.style={->,shorten >=.5pt, shorten <= .5pt, blue},
    nwedge/.style={line width=0.5pt},
    rerouting/.pic={
        \coordinate (i1) at (0,0);
        \node[point=black, label=left:\(j_1\)] (j1) at ($(i1)+(.2,1.9)$) {};
        \node[point=black, label=below:\(j_2\)] (j2) at ($(i1) + (1.8,0.5)$) {};
        \draw[draw=black, line width=1pt] (i1) circle (2.2cm);
    },
    morererouting/.pic={
        \pic{rerouting};
        \node[center=black, label=above left:\(i_1\)] (i1) at (i1) {};
        \node[point=blue, label=above:\(p\)] (p) at ($(i1)+(4,4)$) {};
        \node[center=black, label=below:\(i_2\)] (i2) at ($(p)+(0,-4)$) {};
        \draw[draw=black, line width=1pt] (i2) circle (2.5cm);
    },
}

\DeclarePairedDelimiter{\ceil}{\lceil}{\rceil}
\DeclarePairedDelimiter{\floor}{\lfloor}{\rfloor}
\DeclarePairedDelimiter\abs{\lvert}{\rvert}
\DeclareMathOperator*{\argmin}{\arg\!\min}
\DeclareMathOperator*{\OPT}{OPT}
\DeclareMathOperator*{\C}{\mathscr{C}}
\DeclareMathOperator{\GF}{\text{GF}}
\DeclareMathOperator{\DS}{\text{DS}}
\DeclareMathOperator{\CDS}{\C_{\DS}}
\DeclareMathOperator{\LP}{\text{LP}}
\DeclareMathOperator{\cost}{\text{cost}}

\DeclareMathOperator{\dsfactor}{\gamma}
\DeclareMathOperator{\dsfactorcenter}{\dsfactor_{\text{center}}}
\DeclareMathOperator{\dsfactormed}{\dsfactor_{\text{med}}}
\DeclareMathOperator{\dsfactormeans}{\dsfactor_{\text{means}}}

\newcommand{\dickerson}{Dickerson, Esmaeili, Morgenstern, and Pena}

\newcommand{\jones}{Jones, Nguyen and Nguyen}

\newcommand{\chierichetti}{Chierichetti, Kumar, Lattanzi and Vassilvitskii}

\newcommand{\bercea}{Bercea, Gro{\ss}, Khuller, Kumar, R{\"{o}}sner, Schmidt and Schmidt}
\newcommand{\berceaetal}{Bercea et al.}
\newcommand{\ahmadian}{Ahmadian, Epasto, Kumar and Mahdian}

\newcommand{\bera}{Bera, Chakrabarty, Flores and Negahbani}
\newcommand{\beraetal}{Bera et al.}
\newcommand{\thejaswi}{Thejaswi, Ordozgoiti and Gionis}

\newcommand{\cupdot}{\mathbin{\dot{\cup}}}

\begin{document}

\maketitle

\begin{abstract}
    We study discrete \(k\)-clustering problems in general metric spaces that are constrained by a combination of two different fairness conditions within the \emph{demographic fairness} model.
    Given a metric space \((P,d)\), where every point in \(P\) is equipped with a protected attribute, and a number \(k\), the goal is to partition
    \(P\) into \(k\) clusters with a designated center each, such that a center-based objective function is minimized and the attributes are fairly distributed with respect to the following two fairness concepts:
    1) \emph{group fairness}: We aim for clusters with balanced numbers of attributes by specifying lower and upper bounds for the desired attribute proportions.
    2) \emph{diverse center selection}: Clusters have natural representatives, i.e., their centers. We ask for a balanced set of representatives by specifying the desired number of centers to choose from each attribute. 
    
    \dickerson~\cite{dickerson2023doubly} denote the combination of these two constraints as \emph{doubly constrained fair clustering}.
    They present algorithms whose guarantees depend on the best known approximation factors for either of these problems. Currently, this implies an \(8\)-approximation with a small additive violation on the group fairness constraint. 
    For \(k\)-center, we improve this approximation factor to \(4\) with a small additive violation. This guarantee also depends on the currently best algorithm for DS-fair \(k\)-center given by \jones~\cite{jones2020matching}. For \(k\)-median and \(k\)-means, we propose the first constant-factor approximation algorithms.
    Our algorithms transform a solution that satisfies diverse center selection into a doubly constrained fair clustering using an LP-based approach. 
    Furthermore, our results are generalizable to other center-selection constraints, such as matroid \(k\)-clustering and knapsack constraints.
\end{abstract}

\section{Introduction}
With algorithms being used in a wide range of applications that can have a significant impact on individuals and society, fairness in machine learning has become an increasingly important topic in recent years. 
Much of the work is inspired by the disparate impact effect, which describes how machine learning algorithms may produce biased outcomes that disproportionately affect individual groups based on certain protected attributes.
Clustering is a fundamental task in unsupervised learning, making this issue relevant in this context. 
Arguably, the most popular center-based clustering objectives are \(k\)-center, \(k\)-median, and \(k\)-means. 
For the unconstrained versions in general metric spaces, the best known approximation factors are \(2\) for \(k\)-center \cite{hochbaum1986unified, DBLP:journals/tcs/Gonzalez85}, \(2+\epsilon\) for \(k\)-median \cite{CohenAddad2025kmedian} and \(3+2\sqrt{2}+\epsilon\approx 5.83 \) for \(k\)-means \cite{charikar2025kmeans}.

Directly addressing the disparate impact effect, \emph{demographic fairness} notions aim to ensure that protected attributes are fairly represented in the clusters. 
We say that data points belong to the same group if they share the same protected attribute value. 
Throughout, we will assume that any data point belongs to exactly one group, i.\,e., that the groups are disjoint. We also identify groups with different colors.
Since the work of \chierichetti \cite{chierichetti2017fairlets} initiated the rigorous study of fair clustering, a significant amount of research has emerged on this topic \cite{DBLP:journals/access/ChhabraMM21survey,DBLP:conf/satml/0001EMZ25survey}, spanning a variety of different fairness notions. 
Here, we focus on two specific demographic fairness constraints: \emph{group fairness} and \emph{diverse center selection}.
Our goal will be to construct clusterings that satisfy the \emph{combination} of these fairness constraints, but first, we will discuss them individually.
\subparagraph*{Group fairness}
Group fairness concerns the composition of clusters, namely ensuring that the proportion of points from any group in each cluster is within certain bounds. For group fair \(k\)-center, \(k\)-median, and \(k\)-means there exist: a \(3\)-approximation with additive violation of at most \(1\) \cite{bercea2019cost}, a \((4+\epsilon)\)-approximation with violation of \(\le 1\) \cite{bercea2019cost, DBLP:conf/nips/BeraCFN19}, and a \(10.66\)-approximation with violation of \(\le 3\) \cite{DBLP:conf/nips/BeraCFN19}, respectively.
\subparagraph*{Diverse center selection}
Complementing the group fairness notion, diverse center selection aims to ensure fair representation of the groups. As centers are naturally interpreted as representatives of the clusters, this constraint requires that from each group, a certain number of points be selected as centers. We also say that a clustering is \emph{diversity-aware} when the selected centers satisfy this constraint.
For diversity-aware \(k\)-center, \(k\)-median, and \(k\)-means with exact numbers of points per center, there exist: a \(3\)-approximation \cite{jones2020matching}, a \(7.081\)-approximation \cite{thejaswi2021diversity, Krishnaswamy2018Rounding}, and a \(256\)-approximation \cite{thejaswi2021diversity, Vakilian2022indivfair} (and a \(64\)-approximation if the distances are Euclidean \cite{thejaswi2021diversity, zhao2025matroidmeans}), respectively.
\subparagraph*{Doubly constrained fair clustering}
Most works study these two constraints in isolation, but one could argue that, in many applications, it is desirable to achieve both constraints simultaneously. Group fairness alone can be useful to guarantee a range of different opinions and life experiences in a working environment or in a voting group, i.\,e., a cluster. This can boost productivity or prevent manipulation, such as gerrymandering, but loses its effect if the chosen representatives of the groups are not as diverse and act in their own interests. At the other extreme, selecting only diverse centers without addressing the distributions within the clusters leaves the possibility of manipulation or peer pressure, mitigating the advantages of diverse centers.

\dickerson~\cite{dickerson2023doubly} were the first to study the combination of these two constraints for the \(k\)-center problem and coined the term \emph{doubly constrained fair clustering}. 
They present algorithms that sequentially solve the two problems to achieve solutions for the doubly constrained problem. 
The guarantees depend on the order in which the two problems are considered. 
By first computing a group fair solution, and transforming it into a doubly fair solution, they get a \(2\alpha_{\GF}\)-approximation for doubly fair \(k\)-center, where \(\alpha_{\GF}\) is the approximation guarantee of the algorithm used to compute the group fair clustering. Note that their approach requires an algorithm for group fair clustering with no additive violation, which currently does not exist for \(k\)-center. Therefore, we will only compare to their results for the other direction, where they first compute an \(\alpha_{\DS}\)-approximation for diversity-aware \(k\)-center and then transform it into a doubly constrained fair solution, which yields a \(2(\alpha_{\DS}+1)\)-approximation overall.
While one could easily arrange for diverse center selection by moving centers into their own singleton clusters, to ensure group fairness, clusters usually need to contain multiple points in addition to their selected center.

\subparagraph*{Our contributions}
We improve on known results for \(k\)-center and construct the first constant-factor approximation algorithms for doubly constrained fair \(k\)-median and \(k\)-means. 
\begin{itemize}
    \item We improve the currently best known approximation guarantee for doubly constrained fair \(k\)-center \cite{dickerson2023doubly} by a factor of \(2\). 
    Using the currently best approximation factor for diversity-aware \(k\)-center \cite{jones2020matching}, this implies a \(4\)-approximation for doubly constrained fair \(k\)-center with an additive violation of at most \(2\) for the group fairness constraint.
    \item We present the first polynomial-time constant-factor approximation algorithms for doubly constrained fair \(k\)-median and \(k\)-means. For \(k\)-median, our algorithm currently implies a \(10.081\)-approximation, and for \(k\)-means, it yields a \(291+2\sqrt{290}\approx 325.06\)-approximation, both with an additive violation of at most \(2\) for the group fairness constraint.
    \item Our algorithms generalize to combinations of the group fairness constraint with other constraints on the centers, such as matroid and knapsack constraints or individual fairness. 
\end{itemize}

\subparagraph*{Techniques}
Our algorithms for all three objectives follow the same framework:

First, we independently compute a diversity-aware center set using an external approximation algorithm and a (fractional) group fair solution by solving a linear program. 
Next, we combine the two solutions by rerouting assignments so that only centers from the diversity-aware clustering get assigned points. 
Finally, we compute an integral solution by solving a max flow instance for \(k\)-center and a min cost flow instance for \(k\)-median and \(k\)-means. 

The idea for the overall framework and the rerouting of assignments is inspired by the work of \ahmadian~\cite{ahmadian2019Overrepresentation}, who study a special case of group fairness with only upper bounds on the ratios of colors and provide a \(3\)-approximation with additive violation of at most \(2\) for \(k\)-center without over-representation. 
We significantly adapt their techniques to handle our doubly constrained setting.

First, we must incorporate lower bounds into the group fairness constraint. While adapting the linear program is straightforward, it requires a more complex rerouting of assignments to ensure that the lower bounds are still satisfied after reassigning the points.

Second, we need to handle the fact that we have to combine two different solutions. They compute a set of centers that have pairwise distances of at least some appropriate threshold, which guarantees that there are at most \(k\) centers. 
The LP solution guarantees that every point has a close LP center, and the lower bound on the pairwise distances of the centers guarantees that every new center has a close LP center that is not as close to other centers. 
This allows them to reroute all the mass assigned to an old center in the LP solution to the closest new center, and every new center gets assigned some positive mass. 
In our case, we do not have such a guarantee, as the diversity-aware solution does not provide a lower bound on the pairwise distances between the centers. Just reassigning the mass in the same way would not work as some centers from the diversity-aware solution might not get assigned any points, violating the requirement that all centers should be active.
Therefore, we need to reroute the mass in a more complex way.
We ensure that every cluster is assigned some mass, which may require splitting the incoming mass of an LP center appropriately among multiple centers from the diversity-aware solution while maintaining group fairness. Then, we reassign the remaining mass by the approach described above.
For the other objectives, we need to refine the rerouting approach further to keep the solution's cost small. %\ah{maybe we can say something more specific here?} \jh{i think its good as is}

We compare our results with those of \cite{dickerson2023doubly} for the direction in which they transform a diversity-aware solution into a doubly constrained fair solution. 
Our algorithm performs the transformation in the same direction, and no applicable algorithm exists for the other direction.
Their approach is similar to the framework described above, but differs in that they incorporate the centers computed by the diversity-aware clustering algorithm directly into their linear program for group fairness. 
Solving the linear program with the given centers causes the problem that some centers might not get assigned any points, i.\,e., they are not active. To overcome this issue, they replace inactive centers with points from suitable groups inside clusters that are large enough to be split into two group fair clusters. The additional factor of \(2\) in their approximation guarantee comes from this splitting step, which is not necessary in our approach. 
They finalize their solution by rounding the assignments by computing a max-flow, while we need to reroute assignments to other centers first before solving a max-flow instance.

\subsection{Problem definition}

Let $(P,d)$ be a general metric space with $\vert P \vert = n$ and $k \in \{1,\ldots,n\}$ an integer. 
A $k$-clustering is a partition of $P$ into $k$ clusters $C_1,\dots, C_k$ optimizing some objective function.
We are considering center-based objectives, where a clustering is induced by a set of centers \(\C\subseteq P\) with \(|\C|\le k\) and an assignment \(\varphi\colon P\to \C\) mapping points to centers. 
Note that we study the \emph{discrete} variant of \(k\)-clustering, where the set of centers must be a subset of the point set.
The cluster with center \(c\) is given by \(C = \varphi^{-1}(c)\). 
In the \(k\)-center objective, our goal is to minimize the maximal distance of any point to its assigned center, i.\,e., we want to minimize the maximal radius \(\max_{p\in P}d(p,\varphi(p))\) among the clusters.
In \(k\)-median, the objective is to minimize the sum of distances of points to their assigned center, i.\,e., \(\sum_{p\in P}d(p,\varphi(p))\),
and in \(k\)-means it is the sum of squared distances of points to centers that we want to minimize, i.\,e., \(\sum_{p\in P}d(p,\varphi(p))^2\).

We study clusterings that satisfy two types of demographic fairness constraints. Both fairness notions rely on a coloring of the points.
Let \(m\) be some number of colors. For \(1\le h \le m\), we denote the group of points of color \(h\) by \(P_h\). We assume that every point in \(P\) is assigned exactly one color \(1\le h \le m\), i.e., \(P= P_1 \cupdot \ldots \cupdot P_m\). 
On the one hand, we want to achieve a fair distribution of colors in the clusters. This means that every color should constitute a certain fraction in every cluster. For every color \(h\), we are given a lower bound \(\ell_h\) and an upper bound \(u_h\) on these fractions. 
To relax this problem, we further allow a small additive violation of these fractional bounds.
\begin{definition}[Group fair cluster]\label{def:groupfairwithviolation}
    For all \(1\le h\le m\), let \(\ell_h,u_h\) be values with \(0 \le \ell_h \leq u_h \leq 1 \).
    Let \(\rho\in \mathbb{N}\). We say that a cluster \(C\) is \emph{group fair with additive violation \(\rho\)} if 
    \[ \ell_h|C| - \rho \le |C\cap P_h| \le u_h|C| + \rho \]
    for all colors \(1\le h\le m\). We say that a cluster is \emph{group fair} if \(\rho = 0\).
\end{definition}
We say that a clustering fulfills the \emph{group fairness constraint (GF)} if every cluster is group fair. If the lower and upper bounds on the ratios of the groups equal the respective proportions in the overall dataset, we also speak of \emph{exact preservation of ratios} or short \emph{exact group fairness}.
Additionally, we aim for a fair representation of the colors, formalized by lower and upper bounds on the number of centers per group.
\begin{definition}[Diverse center selection]
    Given \(L_h, U_h \in \mathbb{N}\) for \(h \le m\) such that \(\sum_{h \le m} L_h \le k \le \sum_{h \le m} U_h\), we say that a set of centers \(C\) satisfies the \emph{diverse center selection (DS)} constraint if \(L_h \le |C \cap P_h| \le U_h\) for all \(h \le m\).
\end{definition}
The definition in its most general form is also called \emph{fair-range} clustering; the special case with only lower bounds is also referred to as \emph{minority protection}; we denote the case \(L_h=U_h\) for all \(h\in H\) by \emph{exact diverse center selection}.
All guarantees with fixed approximation guarantees stated in this paper hold for exact diverse center selection; some even hold for more general cases, see also the section on related work. 
Our algorithm does not depend on the specific combinatorial structure of the center set and therefore also works for the generalized version, provided appropriate approximation algorithms are available.
Throughout, we will consider only instances with feasible solutions. 

\subsection{Our results}

We have two main results; the first is a constant factor approximation algorithm for the doubly constrained fair \(k\)-center problem.

\begin{restatable}{theorem}{thmcenter} %[Restatement of \cref{thm:k-center}]
    \label{thm:k-center}
    Let \(\dsfactorcenter\) be the approximation factor of a given approximation algorithm for DS-fair \(k\)-center.
    There exists an algorithm that computes a \((\dsfactorcenter+1)\)-approximation with GF-violation of \(\le 2\) for the doubly constrained fair \(k\)-center problem in polynomial time. 
\end{restatable}
\noindent Using the currently best approximation factor of \(3\) for \(k\)-center with DS by \jones~\cite{jones2020matching}, this yields a \(4\)-approximation.

\begin{corollary}
    There exists an algorithm that computes a \(4\)-approximation for the doubly constrained fair \(k\)-center problem with a GF-violation of at most \(2\).
\end{corollary}

The second result extends the techniques of the $k$-center algorithm to obtain approximation algorithms for the doubly constrained fair \(k\)-median problem and the doubly constrained fair \(k\)-means problem.

\begin{restatable}{theorem}{thmmedmeans}\label{thm:medmeans}
    There exists a polynomial time algorithm that computes a
    \begin{itemize}
        \item \((\dsfactormed+3)\)-approximation with GF-violation of at most \(2\) for the doubly constrained fair \emph{\(k\)-median} problem, where \(\dsfactormed\) is the approximation factor of a given approximation algorithm for DS-fair \(k\)-median. 
        \item \((\sqrt{1+(\sqrt{\dsfactormeans}+1)^2}+1)^2\)-approximation with GF-violation of at most \(2\) for the doubly constrained fair \emph{\(k\)-means} problem, where \(\dsfactormeans\) is the approximation factor of a given approximation algorithm for DS-fair \(k\)-means.
    \end{itemize}
\end{restatable}
\noindent Currently the best algorithm for center fair $k$-median computes a $7.081$-approximation \cite{thejaswi2021diversity,Krishnaswamy2018Rounding} and for $k$-means a $256$-approximation \cite{thejaswi2021diversity,Vakilian2022indivfair}, resulting in the following:

\begin{restatable}{corollary}{cormedmeans}
    There exists an algorithm that computes a
    \begin{itemize}
        \item \(10.081\)-approximation with GF-violation of at most \(2\) for the doubly constrained fair \emph{\(k\)-median} problem. 
        \item \(291+2\sqrt{290}\approx 325.06\)-approximation with GF-violation of at most \(2\) for the doubly constrained fair \emph{\(k\)-means} problem.
    \end{itemize}
\end{restatable}
\noindent For the special case of doubly constrained fair \emph{\(k\)-means} under Euclidean Distance we can use the $64$-approximation \cite{thejaswi2021diversity,zhao2025matroidmeans} to achieve an \(83+2\sqrt{82}\approx 101.11\)-approximation with GF-violation of at most \(2\).

We remark that, as we use the approximation algorithms for clustering with diverse centers as opaque-box\footnote{\url{https://www.acm.org/diversity-inclusion/words-matter}} algorithms, our results can also be applied to combine the group fairness constraint with different notions of fairness.
Such a notion would need to apply only to the set of centers, e.\,g., matroid and knapsack constraints \cite{Swamy2016matroid, zhao2025matroidmeans} or the notion of individual fairness \cite{jung2020indivfair,Vakilian2022indivfair}.

\subsection{Related work}
Besides the two specific fairness definitions described above, there is an abundance of different fairness notions. These include, but are not limited to, more constraints focusing on fair distributions of points, such as balance \cite{chierichetti2017fairlets}, bounded representation \cite{ahmadian2019Overrepresentation}, and colorful clustering \cite{DBLP:journals/mp/JiaSS22}. 
Other fairness concepts focus on distances, like individual fairness \cite{DBLP:conf/icml/MahabadiV20} and socially fair clustering \cite{DBLP:conf/fat/AbbasiBV21}.

\chierichetti~\cite{chierichetti2017fairlets} were the first to study the group fairness constraint (albeit not using this name) for the special case of two colors and a specific balance notion. \bercea~\cite{bercea2019cost} presented the first approximation algorithms for group fair clustering for general metric spaces and general bounds on the proportions of groups for an arbitrary number of groups. 
For the fair \(k\)-center problem, there is a deterministic \(3\)-approximation algorithm with additive violation of at most \(1\) \cite{bercea2019cost} and randomized \(3\)-approximation with no additive violation in expectation \cite{DBLP:conf/nips/HarbL20}.
For the exact variant, \berceaetal~\cite{bercea2019cost} provide a \(5\)-approximation algorithm with no additive violation for exact preservation of ratios, and show that it is NP-hard to approximate the assignment version of this problem, i.\,e., finding the optimal \(k\)-center assignment for a given set of centers, better than a factor of \(3\).
For \(k\)-median, \berceaetal~\cite{bercea2019cost} and \beraetal~\cite{DBLP:conf/nips/BeraCFN19} both give algorithms implying a \((4+\epsilon)\)-approximation with an additive violation of at most \(1\) for the group fairness constraint using the recent \((2+\epsilon)\)-approximation for \(k\)-median \cite{CohenAddad2025kmedian}.  
For \(k\)-means, a \(5+4\sqrt{2}+\epsilon\approx10.66\)-approximation with violation of \(\le 3\) is achieved by the algorithm from \bera~\cite{DBLP:conf/nips/BeraCFN19}, using the best known approximation for unconstrained \(k\)-means \cite{charikar2025kmeans}.
All these results assume that the groups are disjoint. Approximation algorithms for the generalization of non-disjoint groups also exist \cite{DBLP:conf/nips/BeraCFN19, DBLP:conf/nips/HarbL20}.
Further related for group fairness include the integration of privacy constraints \cite{DBLP:conf/icalp/Rosner018}, the learning-augmented setting \cite{Wu2026LearningFairkCenter}, streaming, and the distributed model \cite{DBLP:conf/www/CeccarelloPP24}.

We now look at the related work for the diverse center selection constraint.
For the \(k\)-center objective, there exists a \(3\)-approximation for exact center diversity \cite{jones2020matching}, and a \(3\)-approximation for a slight variant of diverse center selection that deals with lower and upper bounds on fractions of the actual proportions of a group in the data set \cite{nguyen2022fair}. If color groups are non-disjoint, the problem is inapproximable to any factor, assuming \(P \neq NP\) \cite{thejaswi2021diversity}.
For the most general version of diverse center selection, there exists an FPT \(3\)-approximation that can even handle outliers \cite{gadekar2026tightfptapproximationsfair}.
For \(k\)-median and \(k\)-means, a result for general \(p\)-norm objectives implies \(e^{O(1)}\)-approximations for general diverse center selection. In the exact case, \thejaswi~\cite{thejaswi2021diversity} show that the problem for the \(k\)-median objective can be reduced to the matroid median problem in polynomial time. This result can also be generalized to other objectives, implying a \(7.081\)-approximation for \(k\)-median and a \(256\)-approximation for \(k\)-means with exact diverse center selection, using the best known guarantees for the matroid version of the respective objective \cite{Krishnaswamy2018Rounding, Vakilian2022indivfair}; where the \(k\)-means variant allows for a \(64\)-approximation if the distances are Euclidean \cite{thejaswi2021diversity, zhao2025matroidmeans}. In the FPT-regime, randomized \((1+\frac{2}{e}+\epsilon)\)- and \((1+\frac{8}{e}+\epsilon)\)-approximations for general diversity-aware \(k\)-median and \(k\)-means can be achieved, respectively \cite{thejaswi2025diversityawareclusteringcomputationalcomplexity}. For Euclidean spaces, \((1+\epsilon)\)-approximations are possible for both objectives under the diverse center selection constraint \cite{zhang2024parameterized}.
Further related work in the realm of diverse center selection includes distributed models \cite{DBLP:journals/jbd/CeccarelloPPS23} and the streaming setting \cite{DBLP:journals/jbd/CeccarelloPPS23, DBLP:conf/aaai/GuoLJC26}.

For doubly constrained fair clustering, the landscape of related work is sparser. Only for \(k\)-center, there exists a constant-factor approximation algorithm \cite{dickerson2023doubly}.
For \(k\)-median, there is an FPT-time approximation algorithm by \cite{wu2025parameterized} that achieves an approximation ratio of \(4+\epsilon\) with an additive violation of at most \(1\) for the group fairness constraint.
For general \(p\)-norm objectives, there is a polynomial-time algorithm by \cite{zhang2025doubly} that achieves an approximation ratio of \(O(\Delta^{1/p})\) with an additive violation of at most \(5\), where \(\Delta\) is the maximum cluster size in the computed solution.
\Cref{tab:comparison} summarizes the results for doubly constrained fair clustering in general metric spaces, comparing our results to the state-of-the-art.

Lastly, we provide a brief overview of additional important center constraints that our algorithm can handle.
There exists a \(3\)-approximation for matroid \(k\)-center \cite{chen2016matroid}, a \(7.081\)-approximation for matroid median \cite{Krishnaswamy2018Rounding}, and a \(256\)-approximation for matroid means \cite{Vakilian2022indivfair}, which is improved to \(64\) in Euclidean spaces \cite{zhao2025matroidmeans}.
For knapsack constraints, the factors are \(3\) \cite{hochbaum1986unified}, \(6.387+\epsilon\) \cite{DBLP:journals/mp/GuptaMZ25} and \(1128+\epsilon\) \cite{zhao2025matroidmeans} for \(k\)-center, \(k\)-median and (Euclidean) \(k\)-means.

\begin{table}[t]
    \centering
    \caption{%
        Comparison of approximation algorithms for doubly constrained fair $k$-clustering in general metric spaces. All results hold for exact center diversity; some of them hold for more general cases.
        Here, $\Delta$ denotes the maximum cluster size in the computed solution, and add.
        GF viol.\ denotes the additive violation of the group fairness constraint.
        The approximation ratios for this work use the currently best diversity-aware algorithms as
        subroutines: a $3$-approximation for $k$-center~\cite{jones2020matching},
        a $7.081$-approximation for $k$-median~\cite{thejaswi2021diversity,Krishnaswamy2018Rounding},
        and a $256$-approximation for $k$-means~\cite{thejaswi2021diversity,zhao2025matroidmeans}.
    }
    \label{tab:comparison}
    \begin{tabular}{@{}llccc@{}}
        \toprule
        \textbf{Objective} & \textbf{Reference} & \textbf{Approx.\ ratio} & \textbf{Add. GF viol.} & \textbf{Runtime} \\
        \midrule
        \multirow{2}{*}{$k$-center}
          & Dickerson et al.~\cite{dickerson2023doubly}
            & $8$                       & $\le 3$ & poly \\
          & \textbf{This work}
            & $\mathbf{4}$              & $\le 2$ & poly \\
        \midrule
        \multirow{3}{*}{$k$-median}
          & Zhang et al.~\cite{zhang2025doubly}
            & $O(\Delta)$               & $\le 5$ & poly \\
          & Wu et al.~\cite{wu2025parameterized}
            & $4+\varepsilon$           & $\le 1$ & FPT  \\
          & \textbf{This work}
            & $\mathbf{10.081}$         & $\le 2$ & poly \\
        \midrule
        \multirow{2}{*}{$k$-means}
          & Zhang et al.~\cite{zhang2025doubly}
            & $O(\sqrt{\Delta})$        & $\le 5$ & poly \\
          & \textbf{This work}
            & $\mathbf{291+2\sqrt{290}\approx 325.06}$ & $\le 2$ & poly \\
        \bottomrule
    \end{tabular}
\end{table}

\section{A \texorpdfstring{\(4\)}{4}-approximation for \texorpdfstring{\(k\)}{k}-center with center diversity and group fairness constraints}\label{section:kcenter}

First, we give an overview of the algorithm, followed by a detailed description of the individual steps. 
As an initial step, compute a fractional solution \((x,y)\) to the $k$-center clustering problem with only the group fairness constraint using linear programming, where \(x\) represents a fractional assignment of points to centers, and \(y\) indicates which points receive mass from other points in \(x\), i.e., are used as centers in \((x,y)\). Furthermore, we compute a set of centers \(\CDS\) satisfying center diversity using a known approximation algorithm for this problem.
As a next step, we reroute the mass in \((x,y)\) such that only the points in \(\CDS\) are used as centers. This yields a fractional solution satisfying center diversity and group fairness. 
Lastly, we convert the fractional assignment to an integral one by using a max flow with lower bounds.
To discuss the algorithm in more detail, we divide it into four components, which are later formalized in \Cref{alg:k-center}.
\begin{enumerate}
    \item Compute a set of centers \(\CDS\) that fulfills center diversity
    \item Compute a fractional group fair solution \((x,y)\)
    \item Reroute mass in \(x\) such that only points in \(\CDS\) receive positive mass \(\to (x',y')\)
    \item Compute an integral point assignment \(\to (x'', y'')\)
\end{enumerate}

\subsection{Center diversity}

To guarantee our solution fulfills the DS constraint we use an existing algorithm for diversity-aware $k$-center as an opaque-box algorithm.
In the following, let \(\CDS\) be the returned set of centers that satisfy the DS constraint, and $\dsfactorcenter$ be the approximation factor of the algorithm.
Currently, the best known algorithm for diversity-aware $k$-center is the $3$-approximation by Jones et al. \cite{jones2020matching} with  $O(nk)$ running time. 
If the combination of group fairness with another constraint is desired, a corresponding approximation algorithm can be used instead.

\subsection{Group fairness}
We formulate a linear program to obtain an initial fractional group fair clustering. We use a matrix of variables \(x\) representing the assignment from points to centers, i.e., \(x_{ij}\) denotes the amount of mass sent from point \(j\in P\) to point \(i\in P\), i.\,e., the amount by which a point \(j\) is assigned to a center \(i\). Additionally, we compute a vector of variables \(y\), where \(y_i\) can be interpreted as the amount to which \(i\in P\) is opened as a center.
\begin{align}
    \sum\limits_{i\in P} x_{ij} &= 1 & \forall j\in P \tag{LP-1.1}\label{lp_line1} \\
    x_{ij} &\leq y_i & \forall i,j \in P \tag{LP-1.2}\label{lp_line2} \\
    \sum\limits_{i \in P} y_i & \leq k \tag{LP-1.3}\label{lp_line3} \\
    \sum\limits_{j \in P_h} x_{ij} & \leq u_h \sum\limits_{j \in P} x_{ij} & \forall h \in H, i\in P \tag{LP-1.4}\label{lp_line4}\\
    \sum\limits_{j \in P_h} x_{ij} &\geq \ell_h \sum\limits_{j\in P} x_{ij} & \forall h\in H, i\in P \tag{LP-1.5}\label{lp_line5} \\
    x_{ij} &= 0 & \forall j\in P, i\in P : d(i,j) > \lambda \tag{LP-1.6}\label{lp_line6} \\
    0 \leq x_{ij}&, y_i \leq 1 & \forall i,j \in P \tag{LP-1.7}\label{lp_line7}
\end{align}
We denote the linear program by \(\LP(\lambda)\).
Note that the \(k\)-center objective function is not linear. However, we can control the cost of a feasible solution via constraint (\ref{lp_line6}), which ensures that points are not assigned to centers farther than \(\lambda\) away. 
Combining this with a search over all possible values of \(\lambda\) allows us to drop the objective function.
The group fairness (as described in \Cref{def:groupfairwithviolation}) is enforced by constraints (\ref{lp_line4}) and (\ref{lp_line5}).

This linear program is a correct formalization of \(k\)-center clustering with group fairness constraint. To see this, we prove that any feasible \(k\)-center clustering satisfying group fairness corresponds to a feasible solution for \(\LP(\lambda)\) when \(\lambda\) is at least the objective value of the clustering.
The proof of the following two lemmas can be found in  \cref{appendix:for_sec_3}.

\begin{restatable}[]{lemma}{lemmaLPcontainsallfeasibleclusterings}
\label{lem:LP-contains-all-feasible-clusterings}
    Let \((\C,\varphi)\) be a feasible \(k\)-center clustering satisfying GF with objective value \(\lambda^{\C}\).
    Then, \((x,y)\) defined by 
    \begin{align*}
        x_{ij} =
        \begin{cases}
            1 & \text{if } i=\varphi(j)\\
            0 & \text{otherwise }
        \end{cases} \qquad\text{ and }\qquad
        y_i = 
        \begin{cases}
            1 & \text{if } i\in \C\\
            0 & \text{otherwise }
        \end{cases}
    \end{align*}
    is a feasible solution for \(\LP(\lambda^{\C})\).
\end{restatable}

\begin{restatable}[]{lemma}{lemmaIntegralLPsolutionsareclustering}
\label{lem:Integral-LP-solutions-are-clusterings}
    Let \((x,y)\) be a feasible integral solution for \(\LP(\lambda)\). Then, \((\C,\varphi)\) defined by 
    \[ \C=\{i\in P\mid y_i=1\}\quad \text{ and }\quad \varphi(j) = \arg\max_{i\in P}x_{ij}\ \ \text{ for all } j\in P \]
    is a feasible \(k\)-center clustering satisfying group fairness with objective value \(\lambda\).
\end{restatable}

We choose the value of \(\lambda\) as follows. 
It should be large enough such that \(\LP(\lambda)\) allows for a feasible solution. For this, let \(\lambda^{\LP}\) denote the minimal value such that \(\LP(\lambda^{\LP})\) has a solution. 
Further, \(\lambda\) should be large enough such that the \(\dsfactorcenter\)-approximation for \(k\)-center with DS has cost upper bounded by \(\dsfactorcenter\cdot\lambda\).
For this, let \(\lambda^{\text{DS}}\) be minimal such that the \(\dsfactorcenter\)-approximation returns a solution with exactly \(k\) centers and value of at most \(\dsfactorcenter\cdot\lambda^{\text{DS}}\). 
Now, we set \(\lambda \coloneqq \max\{\lambda^{\LP}, \lambda^{\text{DS}}\}\).
We can upper bound the value of \(\lambda\) by the value of an optimal solution satisfying both GF and DS.
\begin{restatable}[]{lemma}{lemmaFeasibility}
% \begin{lemma}
    Let \(\OPT_{\text{GF+DS}}\) be the value of an optimal \(k\)-center solution satisfying GF and DS.
    Then, \(\lambda\coloneqq \max\{\lambda^{\LP}, \lambda^{\text{DS}}\} \le \OPT_{\text{GF+DS}}\).
\end{restatable}
\begin{proof}
    Let \(\OPT_{\GF}\) and \(\OPT_{\DS}\) denote the values of an optimal \(k\)-center clustering satisfying GF and DS, respectively.
    By \Cref{lem:LP-contains-all-feasible-clusterings}, \(\LP(\OPT_{\GF})\) contains a feasible solution. This implies \(\lambda^{\LP}\le \OPT_{\GF}\).
    Let \(\cost_{\DS}\) denote the cost of the solution computed by the \(\dsfactorcenter\)-approximation for \(k\)-center with DS. It is \(\cost_{\DS}\le \dsfactorcenter\cdot\OPT_{\DS}\).
    We set \(\lambda^{\text{DS}}\) to the minimal value such that \(\cost_{\DS}\le \dsfactorcenter\cdot \lambda^{\DS}\). Then, \(\lambda^{\DS}\le \OPT_{\DS}\).

    Any solution satisfying both GF and DS is also a feasible solution for GF or DS, respectively. Hence, \(\OPT_{\GF} \le \OPT_{\GF+\DS}\) and \(\OPT_{\DS}\le \OPT_{\GF+\DS}\).
    Putting everything together yields \(\max\{\lambda^{LP},\lambda^{\DS}\} \le \max\{\OPT_{\GF},\OPT_{\DS}\} \le \OPT_{\GF+\DS}\).
\end{proof}

\subsection{Rerouting fractional assignments}\label{sec:rerouting}
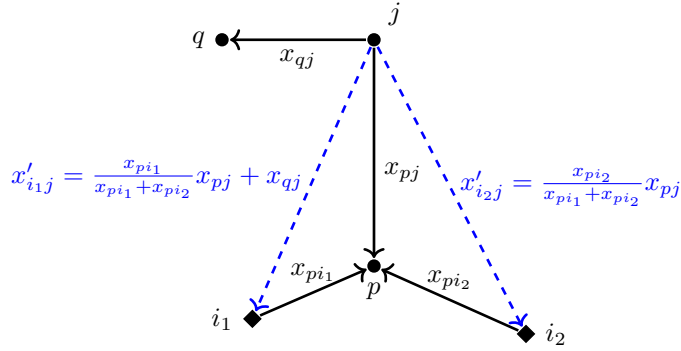
\begin{figure}
    \centering
    % figure showing the reassignment of mass of a center from the DS-solution
\tikzset{
	every picture/.style={line width=1pt},
	point/.style args={#1}{circle, draw=#1, fill=#1, minimum size=4pt, inner sep=0pt},
	center/.style args={#1}{diamond, draw=#1, fill=#1, line width= .5, minimum size=6.5pt, inner sep=0pt},
	outlier/.style args={#1}{circle, draw=#1, minimum size=4pt, inner sep=0pt},
    assign/.style={->,shorten >=.5pt, shorten <= .5pt},
    reassign/.style={dashed, ->,shorten >=.5pt, shorten <= .5pt, blue},
    rerouting/.pic={
        % \coordinate (p) at (0,0);
        \node[point=black, label=below:\(p\)] (p) at (0,0) {};
        \node[center=black, label=left:\(i_1\)] (i1) at ($(p)+(-1.6,-.7)$) {};
        \node[center=black, label=right:\(i_2\)] (i2) at ($(p)+(2,-.9)$) {};
        \node[point=black, label=above right:\(j\)] (j) at ($(p)+(0,3)$) {};
        \node[point=black, label=left:\(q\)] (q) at ($(j)+(-2,0)$) {};
        
        % \draw[draw=gray, line width=1pt] (i1) circle (2.2cm);
        % \draw[draw=gray, line width=1pt] (i2) circle (2.6cm);
        % \node at ($(i1)+(-.2,-1.8)$) {\textcolor{gray}{\(N(i_1)\)}};
        % \node at ($(i2)+(.4,-2)$) {\textcolor{gray}{\(N(i_2)\)}};
    },
}

\begin{tikzpicture}
    \pic{rerouting};
    
    \draw[assign] (i1) to node[above] {\(x_{pi_1}\)}  (p);
    \draw[assign] (i2) to node[above] {\(x_{pi_2}\)}  (p);
    \draw[assign] (j) to node[below right] {\(x_{pj}\)} (p);
    \draw[assign] (j) to node[below] {\(x_{qj}\)} (q);
    \draw[reassign] (j) to node[left] {\(x'_{i_1j} = \frac{x_{pi_1}}{x_{pi_1}+x_{pi_2}}x_{pj}+x_{qj}\)} (i1); 
    \draw[reassign] (j) to node[right] {\(x'_{i_2j} = \frac{x_{pi_2}}{x_{pi_1}+x_{pi_2}}x_{pj}\)} (i2); 
\end{tikzpicture}
    \caption{Rerouting the mass outgoing at a point \(j\). There are two centers \(i_1,i_2\in \CDS\). The point \(j\) covers both cases: there is a point \(p\in N(\CDS)\) and a point \(q\not\in N(\CDS)\) that \(j\) sends mass to. 
    The mass \(x_{pj}\) is split and rerouted to the two centers \(i_1, i_2\in \CDS\) proportionally to the mass they send to \(p\). Center \(i_{c}\) gets an \(\frac{x_{pi_{c}}}{x_{pi_1}+x_{pi_2}}\)-fraction of \(x_{pj}\) for \(c=1,2\). 
    Let us assume that \(i_1=\argmin_{i\in \CDS}d(q,i)\). Then, the mass \(x_{qj}\) is fully rerouted to \(i_1\).
    The original mass assignments are indicated using solid black lines; the rerouted mass assignments are shown with \textcolor{blue}{blue} dashed lines. Omitted edges correspond to \(0\) assignments.}
    \label{fig:rerouting}
\end{figure}
Let \((x,y)\) be the (fractional) group fair solution computed by solving the LP optimally, and \(\CDS\) be the set of diverse centers computed by a \(\dsfactorcenter\)-approximation for \(k\)-center with DS constraint. 
If \(x_{ij}>0\) for some \(j\in P\), the point \(i\) can be considered as a center in this solution. 
We want to reroute mass to the centers in \(\CDS\), such that positive mass is only sent to points in \(\CDS\). Meanwhile, we want to ensure that all centers in \(\CDS\) get some positive mass.
That is, we construct \((x',y')\) such that \(\{i\in P \mid \sum_{j\in P}x'_{ij}>0\} = \CDS\), while ensuring that the cost does not increase too much.
In this context, we need the notion of a point's \emph{neighborhood}. For a point \(j\in P\), we define the set \(N(j) \coloneqq \{i\in P \mid x_{ij} > 0\}\) of all the points that \(j\) sends positive mass to (i.\,e., is assigned to in the fractional solution). For a set of points \(P'\subseteq P\), we define \(N(P') \coloneqq \bigcup_{j\in P'}N(j)\). 

The overall idea is to take all the mass incoming at some point \(p\in P\) in the LP solution \((x,y)\) and reroute it to a ``nearby'' center in \(\CDS\) if this is possible. 
We have to be careful that this does not change the color proportions. Then, as \((x,y)\) fulfills group fairness, the newly formed clusters are also group fair.
Consider a point \(p\in P\) that gets some mass in the current solution, i.e., \(\sum_{j\in P}x_{pj}>0\).
Depending on \(p\)'s position relative to the centers in \(\CDS\), we decide how to reroute the incoming mass at \(p\). 

We want to prevent any center in \(\CDS\) from becoming empty (i.e., it gets no mass assigned).
By the first constraint of the LP, \(\sum_{i\in P}x_{ij}=1\), every point, particularly every center in \(\CDS\), must possess some outgoing mass in the LP solution \((x,y)\). That is, for every \(c\in \CDS\), the set \(N(c)\) is non-empty.
To ensure that mass is sent to a point \(c\in\CDS\), we reroute mass sent to a point \(p\in N(c)\) to \(c\).
However, these sets might overlap, and if we sent all incoming mass at \(p\in N(c)\) to \(c\), there might not be any mass left to send to some other center that has \(p\) in its neighborhood.
To circumvent this problem, we proportionally split the mass sent to a point \(p\in N(\CDS)\) among all points \(c\in \CDS\) with \(p\in N(c)\).
A center \(i\in \CDS\) gets as much mass from \(j\in P\) with \(x_{p,j}>0\) as it sends to \(p\), proportional to the masses sent from all centers \(i'\in \CDS\) to \(p\).

In the remaining case that \(p\not\in N(\CDS)\), \(p\) gets mass solely from points that are not desired centers themselves. In this case, we can safely reroute all mass \(\sum_{j\in P}x_{pj}\) to some nearby center in \(\CDS\). To keep the distances small, we choose the center \(i\in \arg\min_{i'\in \CDS} d(p,i')\) that is closest to \(p\). 
Formally, a new fractional solution \((x',y')\) is defined as follows:
\begin{align}\tag{rerouting-\(y\)}\label{eq:rerouting-y}
    y'_i = \begin{cases}
        1 & \text{if } i\in \CDS\\
        0 & \text{otherwise}
    \end{cases}
\end{align}
and
\begin{align} \tag{rerouting-\(x\)}\label{eq:rerouting-x}
    x'_{ij} =
    \begin{cases}
        \sum_{p\in N(i)} \frac{x_{pi}}{\sum_{c\in \CDS\colon p\in N(c)}x_{pc}}x_{pj} + \sum_{p\in \theta^{-1}(i)}x_{pj} & \text{ if } i\in \CDS\\
        0 & \text{ otherwise,}
    \end{cases}
\end{align}
where \(\theta\colon P\setminus N(\CDS) \to \CDS\) with 
\[ 
\theta(p) = 
\begin{cases}
    p & \text{ if } p\in \CDS\\
    i & \text{ with } i \in \argmin_{i'\in \CDS}d(p,i'), \text{ otherwise}
\end{cases}
\]

\Cref{fig:rerouting} schematically visualizes how the mass emitted from a point \(j\) is rerouted to centers in \(\CDS\). 
We can show that the rerouted solution \((x',y')\) is a feasible solution to \(\LP((\dsfactorcenter+1)\lambda)\) such that \(y'\) is integral. Further, every \(c\in \CDS\) gets a mass of at least \(1\) assigned.
We defer the proof to \cref{appendix:for_sec_3}.

\begin{restatable}[]{lemma}{lemmaFirstRerouting}
\label{lem:rerouting}
    Let \((x,y)\) be a feasible solution for \(\LP(\lambda)\) and let \((x',y')\) be the solution achieved through the rerouting described above. Then the following properties hold:
    \begin{enumerate}
        \item \((x',y')\) is a feasible solution to \(\LP((\dsfactorcenter+1)\lambda)\)
        \item \(\sum_{j\in P}x'_{ij}\ge 1\) for all \(i\in \CDS\)
        \item \(y'\) is integral.
    \end{enumerate}
\end{restatable}

\subsection{Final fair assignment}\label{sec:finalAssign}

\begin{figure}
    \centering
    \begin{tikzpicture}
        \def \scalefactor{0.5}
        \def \breadthfactor{2}
        \node[point] (start) at (0,0) {};
        
        \node[point=blue] (b1) at (\breadthfactor*1,\scalefactor*3) {};
        \node[point=red] (r1) at (\breadthfactor*1,\scalefactor*2) {};
        \node[point=blue] (b2) at (\breadthfactor*1,\scalefactor*1) {};
        \node[point=blue] (b3) at (\breadthfactor*1,\scalefactor*0) {};
        \node[point=red] (r2) at (\breadthfactor*1,-\scalefactor*1) {};
        \node[point=red] (r3) at (\breadthfactor*1,-\scalefactor*2) {};
        \node[point=blue] (b4) at (\breadthfactor*1,-\scalefactor*3) {};

        \node[center=red] (cr1) at (\breadthfactor*2,\scalefactor*2) {};
        \node[center=blue] (cb1) at (\breadthfactor*2,\scalefactor*1) {};
        \node[center=red] (cr2) at (\breadthfactor*2,-\scalefactor*1) {};
        \node[center=blue] (cb2) at (\breadthfactor*2,-\scalefactor*2) {};

        \node[center] (c1) at (\breadthfactor*3,\scalefactor*1.5) {};
        \node[center] (c2) at (\breadthfactor*3,-\scalefactor*1.5) {};

        \node[point] (end) at (\breadthfactor*4,0) {};

        \foreach \p in {b1,b2,b3,b4,r1,r2,r3}
        {
            \draw[nwedge] (start) -- (\p);
        }
        \foreach \p in {b1,b2,b3,b4}
        {
            \draw[nwedge] (\p) -- (cb1);
        }
        \foreach \p in {r1,r3}
        {
            \draw[nwedge] (\p) -- (cr1);
        }
        \foreach \p in {r1,r2,r3}
        {
            \draw[nwedge] (\p) -- (cr2);
        }
        \foreach \p in {b1,b3}
        {
            \draw[nwedge] (\p) -- (cb2);
        }
        \foreach \c in {cb1,cr1}
        {
            \draw[nwedge] (\c) -- (c1);
        }
        \foreach \c in {cb2,cr2}
        {
            \draw[nwedge] (\c) -- (c2);
        }
        \foreach \c in {c1,c2}
        {
            \draw[nwedge] (\c) -- (end);
        }
        \node at ($(start)+(-0.5,0)$) {\(s\)};
        \node at ($(b1)+(0,0.5)$) {\(P\)};
        \node at ($(cr1)+(0,0.5)$) {\(\CDS\times H\)};
        \node at ($(c1)+(0,0.5)$) {\(\CDS\)};
        \node at ($(end)+(0.5,0)$) {\(t\)};
        \node at ($0.5*(start)+0.5*(b4)+(-.3,-.3)$) {\(A_1\)};
        \node at ($0.5*(b4)+0.5*(cb2)+(0,-.3)$) {\(A_2\)};
        \node at ($0.5*(cb2)+0.5*(c2)+(0,-.3)$) {\(A_3\)};
        \node at ($0.5*(c2)+0.5*(end)+(.3,-.3)$) {\(A_4\)};
    \end{tikzpicture}
    \caption{Flow network to find the final assignment of points to centers. Every point in $p \in P$ has a node representing it, which has an edge in $A_1$ to the node $s$. 
    In the figure, these nodes are colored \textcolor{blue}{blue} or \textcolor{red}{red} to represent the color $h_p$ of $p$. 
    Every such node has an edge in $A_2$ to every node corresponding to a pair of a center $c \in \CDS$ and the color $h_p$, i.\,e., $(c, h_p) \in \CDS \times H$, which are colored accordingly in the figure. 
    Each pair of center and color $(c, h_p)$ has an edge in $A_3$ to a node representing the center $c$ in the pair. 
    Each of these nodes corresponding to a center $c \in \CDS$ has an edge in $A_4$ connecting it to the node $t$. }
    \label{fig:maxflow}
\end{figure}
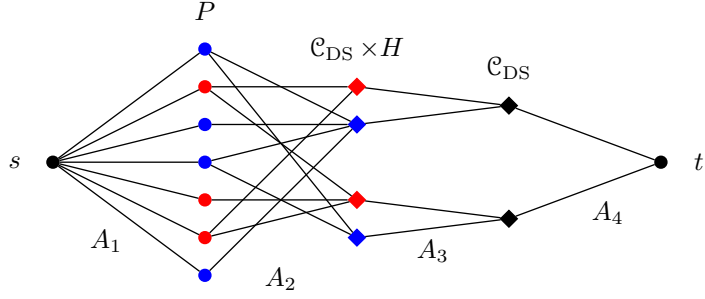
After the rerouting step (see \Cref{sec:rerouting}), we have an assignment that is group fair and satisfies the diverse center selection constraint. However, the assignment \(x'\) is still fractional. To achieve the final integral assignment of points to centers, we build a network graph and solve a max flow problem with lower bounds on it, similar to \cite{ahmadian2019Overrepresentation}. This is done as follows:

Let \(G=(V,E)\), where \(V = \{s,t\} \cup P \cup \{(i,h) \mid i \in \CDS, h \in H \} \cup \CDS\) and \(E = A_1 \cup A_2 \cup A_3 \cup A_4\). 
The set of edges \(A_1= \{(s, j)\mid j \in P\}\) connects all points in \(P\) to the source. Each point \(j \in P\) is connected to a center-color-node \((i,h)\), if \(j \in P_h\) and \(x'_{ij} = 1\). 
More formally \(A_2= \{(j, (i, h)) \mid i\in \CDS, j \in P_h, x'_{ij} > 0\}\). The edges in \(A_1\) and \(A_2\) have capacity 1.
Each of the color-center-nodes \((i,h)\) is connected to its corresponding center \(i\), thus \(A_3= \{((i, h), i)\mid i\in \CDS, h\in H\}\). 
The capacities on these edges ensure that each center receives (nearly) the same amount of flow of each color as it does in the fractional group fair solution. Since \(\sum_{j\in P_h}x^\prime_{ij}\) is not necessarily integral, lower and upper bounds on the edges are introduced; the lower bound is \(\floor{\sum_{j\in P_h}x^\prime_{ij}}\) and the upper bound is \(\ceil{\sum_{j\in P_h}x^\prime_{ij}}\). 
Lastly, the set of edges \(A_4= \{(i, t)\mid i \in \CDS\}\) connects all center points to the sink \(t\). These edges are used to ensure that each cluster center receives (nearly) as much total flow as it received in the fractional solution, but again \(\sum_{j\in P}x^\prime_{ij}\) is not necessarily integral. Each edge has the lower bound \(\floor{\sum_{j\in P}x^\prime_{ij}}\) and the upper bound \(\ceil{\sum_{j\in P}x^\prime_{ij}}\) (see \Cref{fig:maxflow}).

The important property of a max flow we use here is that in a network with integral capacities and integral demands, there always exists a feasible assignment that is integral. Since there is a flow with value $|P|$ (using $x^\prime$), there is also an integral flow with value $|P|$. 
Let $x^{\prime\prime}$ be the integral assignment from the flow. Our final solution is given by the pair $(y^{\prime}, x^{\prime\prime})$. This step can introduce an additive violation of 2.

\begin{restatable}[Group fairness -- \(k\)-center]{lemma}{lemmaAdditiveViolkcenter}
\label{lem:finalAssign}
    For each center \(i\) and each color \(h\) it holds that \[ l_h \sum\limits_{j \in P}x''_{ij} - 2 \leq \sum\limits_{j \in P_h}x''_{ij} \leq u_h \sum\limits_{j \in P}x''_{ij} + 2.\]
\end{restatable}
\begin{proof}[Proof (analogous to \cite{ahmadian2019Overrepresentation})]
The full proof for \Cref{lem:finalAssign} can be found in  \Cref{appendix:for_sec_3}.   
\end{proof}

\begin{algorithm}
    \caption{Doubly constrained fair \(k\)-center}
    \label{alg:k-center}
    \Input{point set $P$, distance metric $d$, $k \in \mathbb{N}$, group fairness bounds $\ell_h, u_h \in \mathbb{N}$ for all \(h\in H\), algorithm for center constraint $\mathcal{A}_{center}$}
    \Output{ center set \(\C \) and assignment \(\varphi\)}
    \(\CDS,\cost_{\DS} \gets \) centers and objective value returned by $\mathcal{A}_{center}(P,d,k)$\\
    \(R\gets \{d(p,q)\mid p,q\in P\}\)\\
    \(\lambda^{\DS} \gets \min\{r\in R\mid r\geq \cost_{\DS}/\dsfactorcenter\} \) \\
    \(\lambda^{\GF} \gets \min\{r\in R\mid \LP(r) \text{ has a solution}\} \) \\
    \(\lambda \gets \max\{\lambda^{\LP}, \lambda^{\DS}\}\) \\
    \((x,y) \gets\) optimal (fractional) solution for \(\LP(\lambda) \) under GF constraints $\ell, u \in \mathbb{N}^{\vert H \vert}$\\
    \((x',y') \gets \) rerouting of \((x,y)\) as specified in Equations (\ref{eq:rerouting-x}) and (\ref{eq:rerouting-y})\\
    \((x'',y'') \gets \) final integral assignment via max flow as described in \Cref{sec:finalAssign} \\
    \(\C\gets \{i\in P\mid y''_i=1\}\)\\
    \For{\(j\in P\)}{
        \(\varphi(j)\gets \arg\max_{i\in P} x''_{ij}\)
    }
    \Return{\((\C,\varphi)\)}
\end{algorithm}
\noindent This construction results in \cref{thm:k-center}. We defer the full proof to \cref{appendix:for_sec_3}.
\thmcenter*

\section{An Algorithm for \texorpdfstring{\(k\)}{k}-median/\texorpdfstring{\(k\)}{k}-means with center diversity and group fairness constraints}\label{section:Kmedmeans}

In this section, we look at the doubly fair clustering problem under the $k$-median and $k$-means objective functions. The basic structure of the algorithm we use for these objectives is the same as for the $k$-center objective described in \Cref{section:kcenter}. We again use a fractional LP solution that fulfills group fairness and reroute it such that it uses centers we obtain via an opaque-box algorithm for center diversity. The resulting fractional assignment is then rounded via a min cost flow procedure.
In the following, we describe the steps of the algorithm for these cost functions and prove its correctness. For notational simplicity we describe the $k$-median cost of a solution $(x,y)$  by $\cost(x,y)=\sum_{i\in P}\sum_{j\in P}x_{ij}d(i, j)$ and the $k$-means cost by $\cost^2(x,y)=\sum_{i\in P}\sum_{j\in P}x_{ij}d(i, j)^2$.

\subsection{Center diversity}\label{sec:DSmedmeans} 
We find a center fair solution \((\CDS,\varphi)\) via an opaque-box algorithm for the respective problem. For notation reasons, we can interpret this solution as an integral assignment \((x^{DS}, y^{DS})\) with \(y^{DS}_i=1\) if and only if \(i \in \CDS\) and \(x^{DS}_{ij}=1\) iff \(\varphi(j)=i\).
Let $\dsfactormed$ and $\dsfactormeans$ be the resulting approximation factors, respectively.
The currently best available approximation algorithms for diverse center selection for these objectives are the following: 
For $k$-median with DS, Thejaswi et al. \cite{thejaswi2021diversity} give a reduction to the matroid median problem, allowing a $\rho$-approximation algorithms for the matroid median problem to be used to compute a $\rho$-approximation for $k$-median with DS.
Currently the best known algorithm for the matroid median problem is the $7.081$-approximation algorithm by Krishnaswamy et al. \cite{Krishnaswamy2018Rounding}.

The reduction by Thejaswi et al. \cite{thejaswi2021diversity} can also be modified to reduce $k$-means with DS to the facility location problem with $\ell_p$-norm cost under matroid constraint, where $p=2$, introduced by Vakilian and Yal\c{c}{\i}ner \cite{Vakilian2022indivfair}.
For this, we only need to set the opening costs of every facility to $0$ in addition to the construction of the $k$-median reduction. 
As the number of opened facilities is bounded by the partition matroids constructed in the reduction, the arguments by Thejaswi et al. apply here as well.
Vakilian and Yal\c{c}{\i}ner \cite{Vakilian2022indivfair} also give a $16^p$ approximation for the facility location problem with $\ell_p$-norm cost under matroid constraint, meaning we can use it to obtain a 256-approximation for $k$-means with DS. 
For the special case of Euclidean Distance we can also use the construction by Thejaswi et al. \cite{thejaswi2021diversity} to reduce from the matroid means problem to $k$-means with DS. 
This allows us to use the $64$-approximation algorithm for matroid means under Euclidean Distance by Zhao et al. \cite{zhao2025matroidmeans}.

\subsection{Group fairness}\label{sec:GFmedmeans}

The following relaxed LP finds a feasible assignment satisfying the GF constraint and minimizes the $k$-median objective function
\begin{align*}
&\text{minimize } &\sum_{i,j\in P} x_{\color{black}{ij}}\cdot d(i, j)\color{red}^2\tag{LP-Med/\textcolor{red}{LP-Means}}&&\\
&\text{subject to }&\sum_{i\in P}x_{ij}&= 1 &\forall j\in P\tag{LP-2.1}\label{1b}\\
&&x_{ij}&\leq y_i	&\forall  i,j\in P\tag{LP-2.2}\label{2b}&\\
&&\sum_{i\in P}y_i&\leq k\tag{LP-2.3}\label{3b}\\
&&{\sum_{x\in P^h}x_{ij}}&{\leq \alpha_h \sum_{x\in P}x_{ij}}&{\forall h \in H, i\in P}\tag{LP-2.4}\label{4b}\\
&&{\sum_{x\in P^h}x_{ij}}&{\geq \beta_h \sum_{x\in P}x_{ij}}&{\forall h \in H, i\in P}\tag{LP-2.5}\label{5b}\\
&&0\leq x_{ij},& y_i \leq1 &\forall i,j\in P\tag{LP-2.6}\label{7b}
\end{align*}

Note that the LP is in large parts identical to the one in section \ref{section:kcenter}. The only difference is that both the $k$-median and $k$-means objective function are linear, so instead of using \Cref{lp_line6} we now add it as proper objective function to the LP. Which objective function (LP-Med or LP-Means) is used depends on if we are interested in the $k$-median or the $k$-means problem. Analogously to \Cref{lem:LP-contains-all-feasible-clusterings} and \Cref{lem:Integral-LP-solutions-are-clusterings} it can be shown how to obtain clustering solutions from the LP and the other way round.

In the following we denote by $(x, y)$ an optimal solution to the above relaxed LP. Since the algorithm largely works the same for both objectives, we often do not specify which one is used except for cases where it is relevant. Note that the cost of $(x, y)$ is a lower bound for the cost of the optimal doubly fair solution. This holds for both the $k$-median and the $k$-means cost function.

\subsection{Rerouting} 

We use the center set $\CDS$ from \Cref{sec:DSmedmeans} and reroute the point assignments of $(x^{LP},y^{LP})$ in a way that the following hold:
\begin{itemize}
    \item the resulting solution ($x^\prime, y^\prime$) is a feasible solution to the LP
    \item for all $i\in \CDS$ it holds $\sum_{j\in P} x_{ij} \geq 1$ 
    \item we can bound the cost of the resulting solution by the following:
        \begin{itemize}
            \item $\cost(x^\prime, y^\prime) \leq 3\cdot\cost(x, y)+\cost(x^{DS}, y^{DS})$ in the $k$-median case
            \item $\cost^2(x^\prime, y^\prime) \leq (1+p^2+(1+\frac{1}{p^2})(2+q^2))\cost^2(x, y) +(1+\frac{1}{p^2})(1+\frac{1}{q^2})\cost^2(x^{DS}, {y^{DS}})$ for $p, q >0$ in the $k$-means case
        \end{itemize}
    \item all $y'_i$ are integral
\end{itemize}

\noindent We first focus on the second requirement. For $i \in \CDS$ we look at the set $N(i):=\{p\in P\mid x_{p i}>0\}$ which is the set of points that get assigned positive mass from $i$ in $(x, y)$. The main idea is to redirect the mass that is assigned to a point $p\in N(i)$ such that $i$ gets assigned a share of this mass that is equivalent to $x_{pi}$. After we have done this for all $i, p$, we reroute the remaining mass that arrives at $p$ by sending it to the nearest neighbor of $p$ in $\CDS$. For this we use the same mapping function $\theta:\{i \in P\mid y_i>0\}\to \CDS$ as in \Cref{sec:rerouting}.

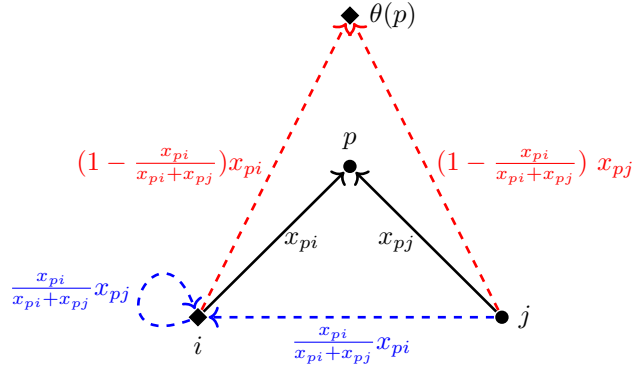
\begin{figure}
        \centering
\tikzset{
	every picture/.style={line width=1pt},
	point/.style args={#1}{circle, draw=#1, fill=#1, minimum size=4pt, inner sep=0pt},
	center/.style args={#1}{diamond, draw=#1, fill=#1, line width= .5, minimum size=6.5pt, inner sep=0pt},
	outlier/.style args={#1}{circle, draw=#1, minimum size=4pt, inner sep=0pt},
    assign/.style={->,shorten >=.5pt, shorten <= .5pt},
    reassign/.style={dashed, ->,shorten >=.5pt, shorten <= .5pt, blue},
    rerouting/.pic={
        \node[point=black, label=above:\(p\)] (p) at (0,0) {};
        \node[center=black, label=below:\(i\)] (i1) at ($(p)+(-2,-2)$) {};
        \node[point=black, label=right:\(j\)] (i2) at ($(p)+(2,-2)$) {};
        \node[center=black, label= right:\(\theta(p)\)] (j) at ($(p)+(0,2)$) {};
    },
}

\begin{tikzpicture}
    \pic{rerouting};
    
    \draw[assign] (i1) to node[right] {\(x_{pi}\)}  (p);
    \draw[assign] (i2) to node[left] {\(x_{pj}\)}  (p);

    \draw[reassign] (i2) to node[below] {\(\frac{x_{pi}}{x_{pi}+x_{pj}}x_{pi}\)} (i1);
    \draw[reassign,in=110,out=200,, loop , min distance= 40pt] (i1) to node[left] {\(\frac{x_{pi}}{x_{pi}+x_{pj}}x_{pj}\)} (i1); 
    \draw[reassign, color = red] (i2) to node[right] {\((1-\frac{x_{pi}}{x_{pi}+x_{pj}})\ x_{pj}\)} (j);
    \draw[reassign, color = red] (i1) to node[left] {\((1-\frac{x_{pi}}{x_{pi}+x_{pj}}) x_{pi}\)} (j); 
\end{tikzpicture}

        \caption{Rerouting of the mass arriving on $p$. The solution $(x,y)$ is depicted by the black arrows. The \textcolor{blue}{blue} dashed arrows show the rerouting back to $i\in \CDS$ from $p\in N(i)$. The \textcolor{red}{red} dashed arrows show the nearest neighbor rerouting of the remaining mass. }
        \label{fig:rrmed}
\end{figure}

We denote by $r^{i}_{p}:=\frac{x_{p i}}{\sum_{l\in P}x_{p l}}$ the ratio of $x_{pi}$ to the total mass of assignments on $p$. Note that $r^{i}_{p} = 0$ if $x_{pi}= 0$.
Using this notation we define $(x', y')$ the following way:

\[y^\prime_i = \begin{cases}1&\text{if }i\in \CDS\\0&\text{otherwise}\end{cases}%\] 
\text{ and } 
x^\prime_{ij} = \begin{cases}\sum_{p\in N(i)}r^{i}_{p}x_{p j}+ \sum_{p\in \theta^{-1}(i)}(1-\sum_{l^\prime\in \CDS}r^{l'}_{p}) x_{p j}&\text{if }i\in \CDS\\0&\text{otherwise}\end{cases}\]
 An example for this rerouting can be seen in \Cref{fig:rrmed}.
The main difference of this rerouting step to the rerouting in \Cref{sec:rerouting} is that a point $i\in \CDS$ now does not get all the mass that is assigned to a point $p\in N(i)$ (even if he is the only center sending mass there), but only a fraction of it that is equivalent to $x_{pi}$.  This property will later be important when we are bounding the cost of this solution.  

The following lemmas show that $(x', y')$ do indeed have the desired properties. We defer the corresponding proofs to \cref{appendix:for_sec_4}. Note that the solutions $(x', y')$ depend on the objective function ($k$-median or $k$-means), since it is used in the LP in \Cref{sec:GFmedmeans} and the respective opaque-box algorithm in \Cref{sec:DSmedmeans}.
\begin{restatable}{lemma}{lemmaMedMeansFeasability}\label{lemma:medmeans-feasability}
	$(x^\prime, y^\prime)$ is a feasible solution to the LP in \Cref{sec:GFmedmeans}. It holds $\sum_{j\in P}x'_{ij}\geq 1$ for all $i \in \CDS$. All $y'_i$ in this solution are integral. 
\end{restatable}

\begin{restatable}[Cost $k$-median]{lemma}{lemmaKmedCost}\label{lemma:med cost}
        For the $k$-median cost of $(x', y')$ the following holds:
    \[\text{cost}(x^\prime, y^\prime) \leq 3\cdot\text{cost}(x, y)+\text{cost}(x^{DS}, y^{DS}).\]
\end{restatable}

\begin{restatable}[Cost $k$-means]{lemma}{lemmaKmeansCost}\label{lemma:means cost}For the $k$-means cost of $(x', y')$ the following holds for $p, q >0$:
\[\cost^2(x^\prime, y^\prime) \leq (1+p^2+(1+\frac{1}{p^2})(2+q^2))\cost^2(x, y) +(1+\frac{1}{p^2})(1+\frac{1}{q^2})\cost^2(x^{DS}, {y^{DS}}).
\] 
\end{restatable}

\subsection{Final fair assignment}
In this section, we show how to round the fractional assignments $x'$ to get an integral solution with small additive violation of the group fairness constraint. We show that the resulting solution \((x'',y')\) has the same cost as $(x', y')$.

For this, we use the same min cost flow construction as Bercea et al. use in \cite{bercea2019cost} to round their essentially fair clustering solutions.
Similarly to the construction used in \Cref{sec:finalAssign} the vertex set $V$ is given by $V = \{s,t\} \cup P\cup \{(i, h)\mid i\in \CDS, h\in H\} \cup \CDS$ and the set $E$ of edges is given by $E= A_1\cup A_2 \cup A_3\cup A_4$, where $A_1= \{(s, j)\mid j \in P\}$, $A_2= \{(j, (i,h) \mid j \in P, i \in \CDS, h \in H, x'_{ij} > 0 \}$, $A_3= \{((i, h), i)\mid i\in \CDS, h\in H\}$ and $A_4= \{(i, t)\mid i \in \CDS\}$.
All edges have unit capacity. $A_2$ is the only edge set that has nonzero cost. An edge $(p, (i, h))\in A_2$ has cost $d(p, i)$ in the $k$-median case and cost $d(p, i)^2$ in the $k$-means case. The node balances are defined as follows:
\begin{itemize}
\item $s$ has balance $\abs{P}$
\item All $p\in P$ have balance 0
\item All $(i, h) \in \CDS \times H$ have balance $-(\floor{\sum_{j\in P_h}x^\prime_{ij}})$
\item All $i \in \CDS$ have balance $-(\floor{\sum_{j\in P}x^\prime_{ij}}-\sum_{h\in H}\floor{\sum_{j\in P_h}x^\prime_{ij}})$
\item $t$ has balance $-(\abs{P}-\sum_{i\in \CDS}\floor{\sum_{j\in P}x^\prime_{ij}})$
\end{itemize}

The structure of this min cost flow can be seen in \Cref{fig:maxflow} since it is identical to the structure of the max flow instance used in \Cref{sec:finalAssign}.
Since all capacities and balances are integral, we can find an integral solution to the flow instance that has minimal cost.
We call the assignment we extract from this MinCostflow-solution $x^{\prime\prime}$. We show that it satisfies group fairness with an additive violation of 2:

\begin{restatable}[Group fairness -- \(k\)-median, \(k\)-means]{lemma}{lemmaAdditiveViolkcmedmeans}
\label{lem:finalAssignmedmeans}
    For each center \(i\) and each color \(h\) it holds that \[ l_h \sum\limits_{j \in P}x''_{ij} - 2 \leq \sum\limits_{j \in P_h}x''_{ij} \leq u_h \sum\limits_{j \in P}x''_{ij} + 2.\]
\end{restatable}

The final solution is then obtained via the pair $( x^{\prime\prime}, y^{\prime})$. 
We defer the proof that this results in \cref{thm:medmeans} to \cref{appendix:for_sec_4}.

\thmmedmeans*

\section{Conclusion}
We study doubly constrained fair 
\(k\)-clustering problem, requiring both group fairness within clusters and diverse center selection among representatives.
Via a modular and extendable approach that uses an LP-based rerouting and rounding routine, we obtain a 4-approximation for \(k\)-center and the first constant-factor approximations for \(k\)-median and \(k\)-means in general metrics.
Natural next steps include proving nontrivial lower bounds and eliminating the additive violation. Further, developing non-sequential algorithms that enforce both fairness constraints directly might enable better approximation guarantees.

\newpage
\bibliography{references}
\newpage
\markboth{Constant-Factor Approximations for Doubly Constrained Fair \texorpdfstring{\(k\)}{k}-Clustering}{APPENDIX}
\appendix
\section{Omitted proofs from \texorpdfstring{\Cref{section:kcenter}}{Section 2}}\label{appendix:for_sec_3}

\lemmaLPcontainsallfeasibleclusterings*
\begin{proof}
    (\ref{lp_line7}) holds by definition.
    Let \(j\in P\). As the assignment \(\varphi\colon P\to \C\) is a function, there exists exactly one \(i'\in \C\) such that \(\varphi(j)=i\). By definition of \(x\), this \(i'\) is the only element in \(P\) such that \(x_{i'j}>0\). Hence, \(\sum_{i\in P}x_{ij} = x_{i'j} = 1\), which implies (\ref{lp_line1}).

    To prove (\ref{lp_line2}), we observe that \(x_{ij} = 1\) if and only if \(i=\varphi(j)\), which is only possible if \(i\in \C\) by definition of \(\varphi\). For \(i\in \C\), it is \(y_i=1\) by definition. For \(x_{ij}=0\) the constraint is trivially fulfilled. 
    
    It is \(y_i = 1\) for at most \(k\) points \(i\) as \(|\C|\le k\). This implies (\ref{lp_line3}).

    Let \(i\in P\) and \(h\in H\). Then, \(\sum_{j\in P}x_{ij} = \sum_{j\in P}\mathbbm{1}[i=\varphi(j)] = |\varphi^{-1}(i)|\) and
    \(\sum_{j\in P_h}x_{ij} = \sum_{j\in P_h}\mathbbm{1}[i=\varphi(j)] = |\varphi^{-1}(i)\cap P_h|\).
    By GF, it follows that 
    \(|\varphi^{-1}(i)\cap P_h| \ge \ell_h|\varphi^{-1}(i)|\) and \(|\varphi^{-1}(i)\cap P_h| \le u_h|\varphi^{-1}(i)|\). This yields (\ref{lp_line4}) and (\ref{lp_line5}).

    (\ref{lp_line6}) holds as otherwise, there would be \(i,j\in P\) with \(d(i,j)>\lambda^{\C}\) and \(i=\varphi(j)\). This would contradict the objective value of \((\C,\varphi)\).
\end{proof}

\lemmaIntegralLPsolutionsareclustering*
\begin{proof}
    The LP-constraint (\ref{lp_line1}) together with the assumption that \(x\) is integral implies that there exists exactly one \(i\in P\) such that \(x_{ij}=1\) for all \(j\in P\). This means that for every \(j\in P\), there exists exactly one \(i\in P\) such that \(\varphi(j)=i\) and therefore the assignment \(\varphi\) is well-defined. By (\ref{lp_line3}) and integrality of \(y\), it follows that there are at most \(k\) points \(i\in P\) such that \(y_i=1\) and therefore \(|\C|\le k\). 
    % group fairness
    For \(h\in H\) and \(i\in P\), it is \(\sum_{j\in P}x_{ij} = |\varphi^{-1}(i)|\) and \(\sum_{j\in P_h}x_{ij} = |\varphi^{-1}\cap P_h|\). Hence, (\ref{lp_line4}) and (\ref{lp_line5}) immediately imply group fairness.
    The objective value follows directly from (\ref{lp_line6}).
\end{proof}

\lemmaFirstRerouting*
\begin{proof}\hfill
    \begin{enumerate}
        \item 
        \begin{itemize}
            \item For (\ref{lp_line1}), consider \(j\in P\).
            \begin{align*}
                \sum_{i\in \CDS}x'_{ij} &= \sum_{i\in \CDS}\left( \sum_{p\in N(i)} \frac{x_{pi}}{\sum_{\ell\in \CDS\colon p\in N(\ell)}x_{p\ell}}x_{pj} + \sum_{p\in \theta^{-1}(i)}x_{pj} \right) \\
                &= \sum_{p\in N(\CDS)}\sum_{i\in \CDS\colon p\in N(i)}\frac{x_{pi}}{\sum_{\ell\in \CDS\colon p\in N(\ell)}x_{p\ell}}x_{pj} + \sum_{i\in \CDS}\sum_{p\in \theta^{-1}(i)}x_{pj}\\
                &= \sum_{p\in N(\CDS)}x_{pj} + \sum_{p\in P\setminus N(\CDS)}x_{pj} = \sum_{p\in P}x_{pj} = 1,
            \end{align*}
            where the last equality holds because \((x,y)\) is a feasible solution for the LP and the second-to-last equality holds because \(\theta\) is a function defined on \(P\setminus N(\CDS)\).
            \item By construction, \(0\le y'_i \le 1\) immediately holds for all \(i\in P\). Since the values of \(x\) are non-negative, \(x'_{ij}\ge0\) for all \(i,j\in P\) as well. By (\ref{lp_line1}), it therefore follows that \(x'_{ij}\le 1\) for all \(i,j\in P\). This implies (\ref{lp_line7}).
            \item (\ref{lp_line2}) holds since either \(y'_i=1\) and \(x'_{ij}\le 1\) by (\ref{lp_line7}), or \(y'_i=0\) and \(x'_{ij}=0\) for all \(i,j\in P\).
            \item (\ref{lp_line3}) follows since \(|\CDS|\le k\) and \(y_i=1\) if and only if \(i\in \CDS\).
            \item For (\ref{lp_line4}), fix \(i\in \CDS\) and \(h\in H\). Then,  
            \begin{align*}
                \sum_{j\in P_h} x'_{ij} &= \sum_{j\in P_h}\left( \sum_{p\in N(i)} \frac{x_{pi}}{\sum_{\ell\in \CDS\colon p\in N(\ell)}x_{p\ell}}x_{pj} + \sum_{p\in \theta^{-1}(i)}x_{pj} \right)\\
                &= \sum_{p\in N(i)}\left(\frac{x_{pi}}{\sum_{\ell\in \CDS\colon p\in N(\ell)}x_{p\ell}}\sum_{j\in P_h}x_{pj}\right) + \sum_{p\in \theta^{-1}(i)}\sum_{j\in P_h}x_{pj}\\
                &\le u_h \sum_{p\in N(i)}\left(\frac{x_{pi}}{\sum_{\ell\in \CDS\colon p\in N(\ell)}x_{p\ell}}\sum_{j\in P}x_{pj}\right) + u_h\sum_{p\in \theta^{-1}(i)}\sum_{j\in P}x_{pj}\\
                &= u_h \sum_{j\in P} \left( \sum_{p\in N(i)} \frac{x_{pi}}{\sum_{\ell\in \CDS\colon p\in N(\ell)}x_{p\ell}}x_{pj} + \sum_{p\in \theta^{-1}(i)}x_{pj} \right)\\
                &= u_h\sum_{j\in P}x'_{ij}
            \end{align*}
            (\ref{lp_line5}) follows analogously.
            \item For (\ref{lp_line6}), consider \(i,j\in P\) such that \(x'_{ij}>0\). By definition of \(x'\), it must be that \(i\in \CDS\). Also by definition, one of the summands \(\sum_{p\in N(i)} \frac{x_{pi}}{\sum_{\ell\in \CDS\colon p\in N(\ell)}x_{p\ell}}x_{pj}\) or \(\sum_{p\in \theta^{-1}(i)}x_{pj}\) must be \(>0\).
            If \(\sum_{p\in N(i)} \frac{x_{pi}}{\sum_{\ell\in \CDS\colon p\in N(\ell)}x_{p\ell}}x_{pj} > 0\), then there exists \(p\in N(i)\) such that \(x_{pi}>0\) and \(x_{pj}>0\). As \((x,y)\) is a feasible solution for \(\LP(\lambda)\), (\ref{lp_line6}) implies \(d(p,i)\le \lambda\) and \(d(p,j)\le \lambda\). Hence, by triangle inequality, \(d(i,j)\le 2\lambda\).
            
            If \(\sum_{p\in \theta^{-1}(i)}x_{pj}> 0\), then there exists \(p\in \theta^{-1}(i)\) such that \(x_{pj}>0\).
            Again, by (\ref{lp_line6}), \(d(p,j)\le \lambda\).
            Further, as \(\theta(p)=i\), it follows that \(i\in \argmin_{i'\in \CDS}d(p,i')\).
            By definition of \(\lambda\), it is \(d(p,i) = \min_{i'\in\CDS}d(p,i') \le \dsfactor\cdot\lambda\).
            Plugging everything together yields 
            \(d(i,j)\le d(i,p) + d(p,j) \le (\dsfactorcenter+1)\lambda\).
        \end{itemize}
        \item Let \(i\in \CDS\). Then,
        \begin{align*}
            \sum_{j\in P}x'_{ij} &= \sum_{j\in P}\left( \sum_{p\in N(i)} \frac{x_{pi}}{\sum_{\ell\in \CDS\colon p\in N(\ell)}x_{p\ell}}x_{pj} + \sum_{p\in \theta^{-1}(i)}x_{pj} \right)\\
            &\ge \sum_{j\in P} \sum_{p\in N(i)} \frac{x_{pi}}{\sum_{\ell\in \CDS\colon p\in N(\ell)}x_{p\ell}}x_{pj}\\
            &= \sum_{p\in N(i)} \frac{x_{pi}}{\sum_{\ell\in \CDS\colon p\in N(\ell)}x_{p\ell}}\sum_{j\in P}x_{pj},
        \end{align*}
        where the first inequality holds because \((x,y)\) is a feasible solution for \(\LP(\lambda)\) and therefore satisfies \(x_{ij}\ge 0\) for all \(i,j\in P\) (\ref{lp_line7}). Further, we can argue that \(\sum_{j\in P}x_{pj} \ge \sum_{j\in \CDS\colon p\in N(j)}x_{pj}\) for all \(p\in N(i)\) as all summands are non-negative and we sum over a subset of the summands. Using this bound, the numerator and the denominator cancel out, and we remain with \(\sum_{j\in P}x'_{ij} \ge \sum_{p\in N(i)}x_{pi}\). Because \(x_{pi}=0\) if and only if \(p\not\in N(i)\), we can deduce that \(\sum_{p\in N(i)}x_{pi} = \sum_{p\in P}x_{pi} = 1\) using (\ref{lp_line1}).
        \item For all \(i\in P\), \(y_i\in \{0,1\}\) by definition. 
    \end{enumerate}
\end{proof}

\lemmaAdditiveViolkcenter*
\begin{proof}
    Let \(\text{mass}_h(i, x') = \sum_{j \in P_h}x'_{ij}\) and \(\text{mass}(i,x') = \sum_{j \in P}x'_{ij}\) be the sum of points of color \(h\) that is assigned to center \(i\) and the sum of total points assigned to center \(i\) after the rerouting step respectively. As \((x', y')\) is a feasible solution for \(\LP(\lambda)\) (see \Cref{lem:rerouting}), we know that \(l_h \cdot \text{mass}(i,x') \leq \text{mass}_h(i, x') \leq u_h \cdot \text{mass}(i,x')\).  
    Let \(\text{mass}_h(i,x'') = \sum_{j \in P_h}x''_{ij} \) and \(\text{mass}(i,x'') = \sum{j \in P}x''_{ij} \) be the sum of points (of color \(h\)) that is assigned to center \(i\) in in the final point assignment step (see \Cref{sec:finalAssign}). 
    We know that \(\floor{\text{mass}_h(i, x')} \leq \text{mass}_h(i,x'') \leq \ceil{\text{mass}_h(i, x')}\) because the lower and upper bounds on each edge \(((i,h), i )\in A_3\) are \(\floor{\sum_{j\in P_h}x^\prime_{ij}}\) and \(\ceil{\sum_{j\in P_h}x^\prime_{ij}}\). Due to the lower and upper bounds on each edge \((i,t) \in A_4\), which are \(\floor{\sum_{j\in P}x^\prime_{ij}}\) and \(\ceil{\sum_{j\in P}x^\prime_{ij}}\), we additionally know that \(\floor{\text{mass}(i,x')} \leq \text{mass}(i,x'') \leq \ceil{\text{mass}(i,x')}\). Since \(\ceil{\text{mass}_h(i, x')} < \text{mass}_h(i, x') +1\) and \(\floor{\text{mass}_h(i, x')} > \text{mass}_h(i, x') -1\), it is true that \(\text{mass}_h(i,x'') < \text{mass}_h(i, x')+1 \leq u_h \cdot \text{mass}(i,x') + 1 \leq u_h \cdot (\text{mass}(i,x'') + 1) + 1 = u_h \cdot \text{mass}(i,x'') + u_h + 1 \leq u_h \cdot \text{mass}(i,x'')  + 2\) and \(\text{mass}_h(i,x'') > \text{mass}_h(i, x') - 1 \geq l_h \cdot (\text{mass}(i,x'') - 1) - 1 = l_h \cdot \text{mass}(i,x') - l_h - 1 \geq l_h \text{mass}(i,x'') - 2\). 
\end{proof}

\thmcenter*
\begin{proof}
As it follows from Lemmas \ref{lem:Integral-LP-solutions-are-clusterings}, \ref{lem:rerouting} and \ref{lem:finalAssign} that the returned solution of Algorithm \ref{alg:k-center} has cost of at most $4\cdot\OPT$ and satisfies DS exactly and GF with an additive violation of at most \(2\), all that is left to show is the polynomial running time.

Assume that we are given an algorithm $\mathcal{A}_{center}$ for the center constraint with running time $t_{center}$ and can solve an LP in $t_{LP}$ time.
If using the DS $k$-center algorithm by Jones et al. \cite{jones2020matching} as $\mathcal{A}_{center}$, computing a set of diverse centers takes $O(nk)$ time.
To find a good upper bound $\lambda$ for the clustering cost, we iterate over the set of pairwise distances in $P$ and compute an LP solution in time $O(n^2t_{LP})$.
Solving $LP(\lambda)$ takes $t_{LP}$ time.
Rerouting one point requires iterating over all centers for each neighbor of each center $i \in \CDS$ as well as iterating over points not neighboring a center, taking $O(k^2n)$ time per point.
Computing the final assignment first requires building the flow network by iterating over $P$, $\CDS \times H$, $\CDS$, and $P\times H$, requiring $O(nm + km)$ time.
Computing the max flow of this network can be done in polynomial time, e.\,g., using the algorithm by Orlin \cite{orlin2013maxflow} taking $O(VE)$ time on our network, where $V \in O( n+km+k)$  and $E \in O(n +nmk+ mk+ k)$.
Lastly, computing the final assignment takes $O(n^2)$ time.
In total this results in polynomial running time in $O(t_{center}+ n^2t_{LP} +k^2 m^2 n + k m n^2)$.
\end{proof}

\section{Omitted proofs from \texorpdfstring{\Cref{section:Kmedmeans}}{Section 3}}\label{appendix:for_sec_4}
\lemmaMedMeansFeasability*
\begin{proof}
		We first prove the feasibility:
		\begin{itemize}
		\item To show Equation (\ref{1b}) we consider the following: Let $j\in P$. Then
        \begin{align*}
            \sum_{i\in \CDS}x'_{ij} &=  \sum_{i\in \CDS}\left(\sum_{p\in N(i)}r^{i}_{p}\cdot x_{p j}+ \sum_{p\in \theta^{-1}(i)}(1-\sum_{l^\prime\in \CDS}r^{l'}_{p})\cdot x_{p j}\right)\\
            &=  \sum_{i\in \CDS}\sum_{p\in N(i)}r^{i}_{p}\cdot x_{p j}+ \sum_{i\in \CDS}\sum_{p\in \theta^{-1}(i)}x_{p j}-\sum_{i\in \CDS}\sum_{p\in \theta^{-1}(i)}\sum_{l^\prime\in \CDS}r^{l'}_{p}\cdot x_{p j}\\
            &=  \sum_{i\in \CDS}\sum_{p\in N(i)}r^{i}_{p}\cdot x_{p j}+ \sum_{p\in P}x_{p j}-\sum_{p\in P}\sum_{l^\prime\in \CDS}r^{l'}_{p}\cdot x_{p j}\\
            &=  \sum_{i\in \CDS}\sum_{p\in N(i)}r^{i}_{p}\cdot x_{p j}+ \sum_{p\in P}x_{p j}-\sum_{l^\prime\in \CDS}\sum_{p\in N(l')}r^{l'}_{p}\cdot x_{p j}\\
            &= \sum_{p\in P}x_{p j}=1
        \end{align*}

		\item Equation (\ref{2b}) holds automatically since either $y'_i=1$ and $x'_{ij} \leq 1$ by 
        Equation (\ref{1b}) or $y'_i= 0$ and $x'_{ij}=0$ for all $i, j$
		\item Equation (\ref{3b}) follows since \(|\CDS|\le k\) and \(y_i=1\) if and only if \(i\in \CDS\).
		\item to show Equation (\ref{4b}) we fix $i\in \CDS, h\in H$. Then it holds:
		\begin{align*}
		&\sum_{j\in P^h}x'_{ij}\\
        &=\sum_{j\in  P^h}\left(\sum_{p\in N(i)}r^{i}_{p}\cdot x_{p j}+ \sum_{p\in \theta^{-1}(i)}(1-\sum_{l^\prime\in \CDS:p\in N(l')}r^{l'}_{p})\cdot x_{p j}\right)\\
        &=\sum_{p\in N(i)}r^{i}_{p}\cdot\sum_{j\in  P^h}x_{p j}+ \sum_{p\in \theta^{-1}(i)}(1-\sum_{l^\prime\in \CDS:p\in N(l')}r^{l'}_{p})\cdot\sum_{j\in  P^h} x_{p j}\\
        &\leq\alpha_h\sum_{p\in N(i)}r^{i}_{p}\cdot\sum_{j\in  P}x_{p j}+ \alpha_h\sum_{p\in \theta^{-1}(i)}(1-\sum_{l^\prime\in \CDS:p\in N(l')}r^{l'}_{p})\cdot\sum_{j\in  P} x_{p j}\\
        &=\alpha_h\sum_{j\in  P}\left(\sum_{p\in N(i)}r^{i}_{p}\cdot x_{p j}+ \sum_{p\in \theta^{-1}(i)}(1-\sum_{l^\prime\in \CDS:p\in N(l')}r^{l'}_{p})\cdot x_{p j}\right)\\
		&= \alpha_h\sum_{j\in P}x'_{ij}
		\end{align*}

		\item Equation (\ref{5b}) works analogously to Equation (\ref{4b}).
		\end{itemize}

		We now only have to show $\sum_{j\in P}x'_{ij}\geq 1$ for all $i \in \CDS$. 
		\begin{align*}
		\sum_{j\in P}x^\prime_{ij}&\geq \sum_{p \in N(i)} \frac{x_{p i}}{\sum_{l\in P}x_{p l}}\cdot \sum_{j\in P}x_{p j}\\
		&= \sum_{p \in N(i)} x_{p i}\\
		&= 1 
		\end{align*}
\end{proof}

\lemmaKmedCost*

\begin{proof}
        In the following, we will bound the cost of the resulting solution.
        For better readability, we divide the assignment function into two parts:\begin{itemize}
            \item $x'^N$ which describes the part of the assignments that centers get assigned from the mass which their neighbors receive in $(x, y)$. More formal \[x'^N_{ij}=\sum_{p\in N(i)}r^{i}_{p}\cdot x_{pj}\]
            \item $x'^\theta$ which describes the remaining assignments via the nearest neighbor function $\theta$ which leads to \[x'^\theta_{ij}=\sum_{p\in \theta^{-1}(i)}(1-\sum_{l^\prime\in \CDS}r^{l'}_{p})x_{pj}\]
        \end{itemize}
        For an illustration, we defer to \Cref{fig:rrmed} which shows $x'^N$ via the \textcolor{blue}{blue} arrows and $x'^\theta$ via the \textcolor{red}{red} arrows. It is easy to see that it holds $x'_{ij}=x'^N_{ij}+ x'^\Theta_{ij}$. Therefore we know cost$(x', y')=$cost$(x'^N, y')+$cost$(x'^\Theta, y')$.

        We first bound the cost of $(x'^N, y')$. By repeated application of the triangle inequality, we get:
      
        \begin{align}
        &\text{cost}(x'^N, y')=\sum_{i\in \CDS}\sum_{j\in P}x'^N_{ij}d(i, j)\\
        &=\sum_{i\in \CDS}\left(\sum_{p\in N(i)}r^{i}_{p}\cdot\sum_{j\in P}x_{pj}d(i, j)\right)\\
        &\leq \sum_{i\in \CDS}\left(\sum_{p\in N(i)}\frac{x_{pi}}{\sum_{l\in P}x_{pl}}\cdot\sum_{j\in P}x_{pj}(d(p, j)+d(p, i)) \right)\\
        &=\sum_{i\in \CDS}\left(\sum_{p\in N(i)}x_{pi}d(p, i)+\frac{x_{pi}}{\sum_{l\in P}x_{pl}}\cdot\sum_{j\in P}x_{pj}d(p,j) \right)\\
        &\leq \text{cost}(x, y)+ \sum_{i\in \CDS}\sum_{p\in N(i)}\frac{x_{pi}}{\sum_{l\in P}x_{pl}}\cdot\sum_{j\in P}x_{pj}d(p,j)\\
        &=\text{cost}(x, y)+ \sum_{i\in \CDS}\sum_{p\in N(i)}r^{i}_{p}\cdot\sum_{j\in P}x_{pj}d(p,j)\\
        \end{align}
        
        We can bound the cost of $(x'^\theta, y')$ the following way:
        Let $j\in P, i\in \CDS$ and let $\varphi^{DS}(j)=\argmin_{x\in \CDS}d(x, j)$ be the center $j$ is assigned to in the center fair solution $(x^{DS}, y^{DS})$. Then repeatedly using the triangle inequality and the definition of $\theta$ leads to:
        
        \begin{align}
        &x'^\theta_{ij}d(i, j)\\
        &=\sum_{p\in \theta^{-1}(i)}(1-\sum_{l^\prime\in \CDS}r^{l'}_{p})\cdot x_{pj}d(i, j)\\
        &\leq\sum_{p\in \theta^{-1}(i)}(1-\sum_{l^\prime\in  \CDS}r^{l'}_{p})\cdot x_{pj}(d(p, j)+d(p, i))\\
        &=\sum_{p\in \theta^{-1}(i)}\left(x_{pj}(d(p, j)+d(p, i))-\sum_{l^\prime\in \CDS}r^{l'}_{p}\cdot x_{pj}(d(p, j)+d(p, i))\right)\\
        &\leq\sum_{p\in \theta^{-1}(i)}x_{pj}(2d(p, j)+d(\varphi^{DS}(j), j))-\sum_{p\in \theta^{-1}(i)}\sum_{l^\prime\in  \CDS}r^{l'}_{p}\cdot x_{pj}d(p, j)
        \end{align}
        When we sum up over all centers and points this leads to the following bound:
        \begin{align}
        &\text{cost}(x'^\theta, y')=\sum_{i\in \CDS}\sum_{j\in P}x'^\theta_{ij}d(i, j)\\
        &\leq\sum_{i\in \CDS}\sum_{j\in P}\left(\sum_{p\in \theta^{-1}(i)}x_{pj}(2d(p, j)+d(\varphi^{DS}(j), j))-\sum_{p\in \theta^{-1}(i)}\sum_{l^\prime\in  \CDS}r^{l'}_{p}\cdot x_{pj}d(p, j)\right)\\
        &\leq 2\text{cost}(x, y) +\text{cost}(x^{DS}, {y^{DS}})-\sum_{i\in \CDS}\sum_{j\in P}\sum_{p\in \theta^{-1}(i)}\sum_{l^\prime\in \CDS}r^{l'}_{p}\cdot x_{pj}d(p, j)\\
        &= 2\text{cost}(x, y) +\text{cost}(x^{DS}, {y^{DS}})-\sum_{j\in P}\sum_{p\in \theta^{-1}[\CDS]}\sum_{l^\prime\in \CDS}r^{l'}_{p}\cdot x_{pj}d(p, j)\\
        &= 2\text{cost}(x, y) +\text{cost}(x^{DS}, {y^{DS}})-\sum_{j\in P}\sum_{l^\prime\in \CDS}\sum_{p\in N(l')}r^{l'}_{p}\cdot x_{pj}d(p, j)
		\end{align}

        Therefore by combining our bounds for cost$(x'^N, y')$ and cost$(x'^\Theta, y')$ we get 
        \begin{align}
        \text{cost}(x'_{ij}, y')&=\text{cost}(x'^N_{ij}, y')+\text{cost}(x'^\theta_{ij}, y')\\
        &\leq \text{cost}(x, y)+ \sum_{i\in \CDS}\sum_{p\in N(i)}r^{i}_{p}\cdot\sum_{j\in P}x_{pj}d(p,j)+ 2\text{cost}(x, y) \\&+\text{cost}(x^{DS}, {y^{DS}})-\sum_{j\in P}\sum_{l^\prime\in \CDS}\sum_{p\in N(l')}r^{l'}_{p}\cdot x_{pj}d(p, j)\\
        &=3\text{cost}(x, y) +\text{cost}(x^{DS}, y^{DS})
        \end{align}
        which concludes the proof.
\end{proof}

\lemmaKmeansCost*

\begin{proof}
 The analysis of the cost of the $k$-means cost function works analogously to the analysis of the $k$-median objective. The main difference is that we cannot simply use the triangle inequality because the distances are squared. Therefore we will be using the more general version $(a+b)^2\leq (1+p^2)\cdot a^2+(1+\frac{1}{p^2})\cdot b^2$ for $p>0$. In the following $(x'^N, y')$ and $(x'^\theta, y')$ are defined as in the proof of the previous \Cref{lemma:med cost}.

        Again, we start by bounding the cost of $(x'^N, y')$:
        \begin{align}
        &\text{cost}(x'^N, y')=\sum_{i\in \CDS}\sum_{j\in P}x'^N_{ij}d(i, j)^2\\
        &=\sum_{i\in \CDS}\sum_{p\in N(i)}r^{i}_{p}\cdot\sum_{j\in P}x_{pj}d(i, j)^2\\
        &\leq \sum_{i\in \CDS}\left(\sum_{p\in N(i)}r^{i}_{p}\cdot\sum_{j\in P}x_{pj}\cdot((1+\frac{1}{p_1^2})d(p, j)^2+(1+p_1^2)d(p, i)^2) \right)\\
        &=\sum_{i\in \CDS}\left(\sum_{p\in N(i)}(1+\frac{1}{p_1^2})x_{pi}d(p, i)^2+r^{i}_{p}\cdot\sum_{j\in P}(1+p_1^2)x_{pj}d(p,j)^2 \right)\\
        &\leq (1+\frac{1}{p_1^2})\text{cost}(x, y)+ \sum_{i\in \CDS}\sum_{p\in N(i)}r^{i}_{p}\cdot\sum_{j\in P}(1+p_1^2)x_{pj}d(p,j)^2
        \end{align}
        for $p_1\in \mathbb{R}_{>0}$.

        We can bound the cost of $(x'^\theta, y')$ the following way:
        Let $j\in P, i\in \CDS$ and let $\varphi^{DS}(j)=\argmin_{x\in \CDS}d(x, j)$ be the center $j$ is assigned to in the center fair solution $(x^{DS}, y^{DS})$. 
        
        \begin{align}
        &x'^\theta_{ij}d(i, j)^2
        =\sum_{p\in \theta^{-1}(i)}(1-\sum_{l^\prime\in \CDS}r^{l'}_{p})x_{pj}d(i, j)^2\\
        &\leq\sum_{p\in \theta^{-1}(i)}(1-\sum_{l^\prime\in \CDS}r^{l'}_{p})x_{pj}\cdot(1+p_2^2)d(p, j)^2+((1+\frac{1}{p_2^2})d(p, i)^2)\\
        &=\sum_{p\in \theta^{-1}(i)}x_{pj}((1+p_2^2)d(p, j)^2+(1+\frac{1}{p_2^2})d(p, i)^2)\\
        &\color{white}{\leq}\color{black}-\sum_{p\in \theta^{-1}(i)}\sum_{l^\prime\in \CDS}r^{l'}_{p}x_{pj}((1+p_2^2)d(p, j)^2+(1+\frac{1}{p_2^2})d(p, i)^2)\\
        &\leq\sum_{p\in \theta^{-1}(i)}x_{pj}(1+p_2^2)d(p, j)^2+((1+\frac{1}{p_2^2})d(p, i)^2)-\sum_{p\in \theta^{-1}(i)}\sum_{l^\prime\in \CDS}r^{l'}_{p}x_{pj}(1+p_2^2)d(p, j)^2\\
        &\leq\sum_{p\in \theta^{-1}(i)}x_{pj}(1+p_2^2)d(p, j)^2+((1+\frac{1}{p_2^2})((1+\frac{1}{q^2})d(\varphi^{DS}(j), j)^2+(1+q^2)d(p, j)^2)\\
        &\color{white}{\leq}\color{black} -\sum_{p\in \theta^{-1}(i)}\sum_{l^\prime\in \CDS}r^{l'}_{p}x_{pj}(1+p_2^2)d(p, j)^2\\
        &=\sum_{p\in \theta^{-1}(i)}(1+p_2^2+(1+\frac{1}{p_2^2})(1+q^2))x_{pj}d(p, j)^2+((1+\frac{1}{p_2^2})(1+\frac{1}{q^2}))x_{pj}d(\varphi^{DS}(j), j)^2\\
        &\color{white}{\leq}\color{black} -\sum_{p\in \theta^{-1}(i)}\sum_{l^\prime\in \CDS}r^{l'}_{p}x_{pj}(1+p_2^2)d(p, j)^2
        \end{align}
        for $p_2, q>0$. 
        
        If we take the sum over all centers $i\in \CDS$ and points $j\in P$ we get
        \begin{align}        
        &\text{cost}(x'^\theta, y')=\sum_{i\in \CDS}\sum_{j\in P}x'^\theta_{ij}d(i, j)^2\\
        &=\sum_{i\in \CDS}\sum_{j\in P}\sum_{p\in \theta^{-1}(i)} \left[ 1+p_2^2+(1+\frac{1}{p_2^2})(1+q^2))x_{pj}d(p, j)^2\right.\\
        &+\left.(1+\frac{1}{p_2^2})(1+\frac{1}{q^2})x_{pj}d(\varphi^{DS}(j), j)^2) \right] \\
        &\color{white}{\leq}\color{black} -\sum_{i\in \CDS}\sum_{j\in P}\sum_{p\in \theta^{-1}(i)}\sum_{l^\prime\in \CDS}r^{l'}_{p}x_{pj}(1+p_2^2)d(p, j)^2\\
        &\leq (1+p_2^2+(1+\frac{1}{p_2^2})(1+q^2))\text{cost}(x, y) +(1+\frac{1}{p_2^2})(1+\frac{1}{q^2})\text{cost}(x^{DS}, {y^{DS}})\\
        &\color{white}{\leq}\color{black}-\sum_{i\in \CDS}\sum_{j\in P}\sum_{p\in \theta^{-1}(i)}\sum_{l^\prime\in \CDS}r^{l'}_{p}x_{pj}(1+p_2^2)d(p, j)^2\\
        &\leq (1+p_2^2+(1+\frac{1}{p_2^2})(1+q^2))\text{cost}(x, y) +(1+\frac{1}{p_2^2})(1+\frac{1}{q^2})\text{cost}(x^{DS}, {y^{DS}})\\
        &\color{white}{\leq}\color{black}-\sum_{p\in \CDS}\sum_{j\in P}\sum_{l^\prime\in \CDS}r^{l'}_{p}x_{pj}(1+p_2^2)d(p, j)^2
		\end{align}
        for $p_2, q\in \mathbb{R}_{>0}$. 
         
        Therefore if we combine the bounds for cost$(x'^N, y')$ and cost$(x'^\Theta, y')$ and set $p_1=p_2=p$ we get
        \begin{align}
        &\text{cost}(x'_{ij}, y')=\text{cost}(x'^N_{ij}, y')+\text{cost}(x'^\theta_{ij}, y')\\
        &\leq (1+\frac{1}{p^2})\text{cost}(x, y)+ \sum_{i\in \CDS}\sum_{p\in N(i)}r^{i}_{p}\cdot\sum_{j\in P}(1+p^2)x_{pj}d(p,j)^2\\
        &\color{white}{\leq}\color{black}+ (1+p^2+(1+\frac{1}{p^2})(1+q^2))\text{cost}(x, y) +(1+\frac{1}{p^2})(1+\frac{1}{q^2})\text{cost}(x^{DS}, {y^{DS}})\\
        &\color{white}{\leq}\color{black}-\sum_{p\in \CDS}\sum_{j\in P}\sum_{l^\prime\in \CDS}r^{l'}_{p}x_{pj}(1+p^2)d(p, j)^2\\
        &=(1+p^2+(1+\frac{1}{p^2})(2+q^2))\text{cost}(x, y) +(1+\frac{1}{p^2})(1+\frac{1}{q^2})\text{cost}(x^{DS}, {y^{DS}})
        \end{align}
        for $p, q >0$. This concludes the proof.
\end{proof}

\lemmaAdditiveViolkcmedmeans*

\begin{proof}
     First, we remember that by Lemma \ref{lemma:medmeans-feasability} $(x', y')$ was a fractional solution satisfying group fairness. To make notation a bit easier, we denote by mass$_h(i,x)= \sum_{j\in P^h}x_{ij}$ the amount of color $h$ points that get assigned to center $i$ in the solution $(x, y)$. Analogously let mass$(i,x)= \sum_{j\in P}x_{ij}$ be the total amount of points assigned to center $i$ in $(x, y)$. It is easy to check that the following hold:
\begin{itemize}
\item $\floor{\text{mass}_h(i,x^\prime)}\leq \text{mass}_h(i,x^{\prime\prime})\leq \ceil{\text{mass}_h(i,x^\prime)}$ by the balance of $(i, h)$ and the fact that the outgoing edge only has capacity 1
\item$\floor{ \text{mass}(i,x^\prime)}\leq  \text{mass}(i,x^{\prime\prime})\leq \ceil{ \text{mass}(i,x^\prime)}$ by the balance of $i$ and the fact that the outgoing edge only has capacity 1

\end{itemize}

Therefore we know 
\begin{align*}
\text{mass}_h(i,x^{\prime\prime})&\leq \text{mass}_h(i,x^\prime)+1
\leq u_h \cdot \text{mass}(i,x^\prime) +1
\leq u_h(\text{mass}(i,x^{\prime\prime})+1)+1\\
&=u_h\cdot \text{mass}(i,x^{\prime\prime})+u_h+1
\leq u_h\cdot \text{mass}(i,x^{\prime\prime})+2
\end{align*}
and analogously
\begin{align*}
\text{mass}_h(i,x^{\prime\prime})&\geq \text{mass}_h(i,x^\prime)-1
\geq l_h \cdot  \text{mass}(i,x^\prime) -1
\geq l_h(\text{mass}(i,x^{\prime\prime})-1)-1\\
&=l_h\cdot \text{mass}(i,x^{\prime\prime})-l_h-1
\geq l_h\cdot \text{mass}(i,x^{\prime\prime})-2
\end{align*}

This shows that there is an additive violation of at most 2 for the (GF) constraint. 
\end{proof}

\thmmedmeans*

\begin{proof}
To see that $(x'', y')$ satisfies the center fairness constraint, we recall that $\CDS$ fulfilled the center fairness constraint. Since we never opened or changed any centers in $\CDS$, we only have to show that every $i\in \CDS$ gets assigned at least one point in $(x'', y')$. Let $i\in \CDS$. From \Cref{lemma:medmeans-feasability} we know that $\sum_{j\in P}x'_{ij}\geq 1$, and therefore either there exists an $h\in H$ with $\sum_{j\in P^h}x'_{ij}\geq 1$ or $\sum_{j\in P^h}x'_{ij}< 1$ for all $h\in H.$ In the first case we know that the balance of the point $(i, h)$ is smaller or equal to $-1$, which means that at least one point has to send (integral) flow in $(x'', y')$ to it, which means there is at least one point assigned to $i$. If $\sum_{j\in P^h}x'_{ij}< 1$ for all $h\in H$, then we know that $-(\floor{\sum_{j\in P}x^\prime_{ij}}-\sum_{h\in H}\floor{\sum_{j\in P_h}x^\prime_{ij}})\leq -1$ and therefore the vertex corresponding to $i$ gets integral flow from at least 1 point in the solution, meaning that at least one point is assigned to $i$. Thus our solution satisfies center fairness.

The group fairness follows directly from \Cref{lem:finalAssignmedmeans}.

For the cost bound, we see that $(x', y')$ can be translated into a valid solution for the min cost flow with the same cost, and therefore the cost of (x', y') is an upper bound for the cost of $(x'', y')$ for both objective functions. The cost bound for $k$-median follows now directly from \Cref{lemma:med cost} and the fact that both $\cost(x, y)$ and $\frac{1}{\dsfactormed}\cost(x^{\text{DS}}, y^{\text{DS}})$ are lower bounds for the cost of the optimal doubly fair solution, since this optimal solution has to be both, group fair and center fair. The bound for the $k$-means cost follows analogously if we set $p^2=\sqrt{1+(1+\sqrt{\dsfactormeans})^2}$ and $q^2=\sqrt{\dsfactormeans}$. 

The running time analysis is very similar to the $k$-center case. 
Computing a set of centers with the desired center based constrained is dependent on the chosen approximation algorithm.
Both the algorithms we discussed in combination with the reduction by Thejaswi et al. \cite{thejaswi2021diversity}, the matroid median algorithm by Krishnaswamy et al. \cite{krishnaswamy2011matroid} for DS-fair $k$-median, and the matroid means algorithm by Zhao et al. \cite{zhao2025matroidmeans} for DS-fair $k$-means have polynomial running times.
Finding an upper bound on the clustering cost and rerouting the points to centers fulfilling the center-based constraint does not take the same time as in the $k$-center case.
The main difference here is solving a min-cost flow instead of a max flow, which can also be computed in polynomial time, e.\,g., using the algorithm by Chen et al. \cite{Chen2025MinCostFlow}.
The total running time is polynomial for both $k$-median and $k$-means.
This concludes the proof of the theorem.
\end{proof}

\newpage
\section{Notation}

\begin{table}[h]
    \centering
    \begin{tabular}{lp{0.7\textwidth}}
        \((P,d)\) & general metric space with point set \(P\) \\
         \(\C\) & set of centers\\
         \(C_1,\ldots,C_k\) & clusters\\
         \(H\) & set of colors \\
         \(m=|H|\) & number of colors\\
         \(P_h\) & sets of points with color \(h\) for \(h \in \{1,\dots,m\}\) \\
         \(\varphi\colon P \to \C\) & assignment of points to centers\\
         \(\ell_h,u_h\) & lower and upper bound for color \(h\) in the group fairness definition \\
         \(L_h, U_h\) & lower and upper bound for color \(h\) in the definition of diverse center selection\\
         \(\CDS\) & center set of the DS-fair solution\\
         \((x,y)\) & solution of the group fairness LP\\
         \((x', y')\) & solution after the rerouting step \\
         \((x'', y'')\) & final, integral doubly fair solution \\
         \(x_{ij}\) & the amount of mass sent from point \(j\) to point \(i\)\\ 
         \(y_i\) & the extent to which point \(i\) is opened (used as a center)\\
         \(\dsfactorcenter,\dsfactormed, \dsfactormeans\)&currently best approximation factor to the (DS) variant of the resp. objective\\
         mass$(i, x)$  &sum over all $x_{ij}$\\
        mass$_h(i, x)$&sum over all $x_{ij}$ where j has color $h$\\
        \(r_p^i\)&ratio of $x_{pi}$ on mass$(i, x)$
    \end{tabular}
    \caption{Notation used throughout the paper}
    \label{tab:notation}
\end{table}

\end{document}